\documentclass[acmtocl]{acmtrans2m}

\usepackage{stmaryrd}
\usepackage{diagrams}
\usepackage{url}
\usepackage{proof}

\newcommand{\unif}{\stackrel{_?}=}
\newcommand{\unifn}{\stackrel{_?}\approx}
\newcommand{\trans}[1]{{\left\llbracket {#1} \right\rrbracket}}
\newcommand{\etrans}[1]{{{\it Eq}( {#1})}}
\newcommand{\transb}[1]{{\left\llbracket #1  \right\rrbracket^{-1}}}
\newcommand{\transn}[1]{\left\llbracket {\sf #1}  \right\rrbracket_\nabla}
\newcommand{\dom}{{\rm Dom}}
\newcommand{\FV}{{\rm FV}}
\newcommand{\Vars}{{\rm Vars}}
\newcommand{\fresh}{\,\#\,}
\newcommand{\infrule}[3]{\infer[\mbox{(#3)}]{#2}{#1}}
\newcommand{\pair}[2]{\langle #1,#2\rangle}
\newcommand{\mdot}{\kern-.35ex\cdot\kern-.35ex}

\newtheorem{theorem}{Theorem}[section]
\newtheorem{definition}[theorem]{Definition}

\newtheorem{lemma}[theorem]{Lemma}
\newtheorem{corollary}[theorem]{Corollary}
\newtheorem{exampleaux}[theorem]{Example}
\newtheorem{remark}[theorem]{Remark}
\newenvironment{example}{\begin{exampleaux}\rm}{\end{exampleaux}}

\title{Nominal Unification\\from a Higher-Order Perspective}

\author{JORDI LEVY\\
Artificial Intelligence Research Institute (IIIA), \\
Spanish Council for Scientific Research (CSIC)
\and 
MATEU VILLARET\\
Departament d'Inform\`atica i Matem\`atica Aplicada (IMA),\\
Universitat de Girona (UdG)
}


\begin{abstract}
  Nominal Logic is a version of first-order logic with equality,
  name-binding, renaming via name-swapping and freshness of
  names. Contrarily to higher-order logic, bindable names, called
  atoms, and instantiable variables are considered as distinct
  entities. Moreover, atoms are capturable by instantiations, breaking
  a fundamental principle of lambda-calculus.  Despite these
  differences, nominal unification can be seen from a higher-order
  perspective. From this view, we show that nominal unification can be
  reduced to a particular fragment of higher-order unification
  problems: Higher-Order Pattern Unification. This reduction proves
  that nominal unification can be decided in quadratic deterministic
  time, using the linear algorithm for Higher-Order Pattern
  Unification. We also prove that the translation preserves most
  generality of unifiers.
\end{abstract}

\begin{document}
\sloppy

\begin{bottomstuff}
This research has been partially founded by the CICYT research project
TIN2007-68005-C04-01/02/03.
\end{bottomstuff}

\maketitle

\section{Introduction}
\label{sec-introduction}

\emph{Nominal Logic} is a version of first-order many-sorted logic
with equality and primitives for renaming via name-swapping,
name-binding, and freshness of names. It is characterized by a
syntactic distinction between \emph{atoms} (that roughly correspond to
the notion of bound variable) and \emph{variables} (that would
correspond to free variables). Therefore, binders can only bind atoms,
we can only instantiate variables, and atoms are not instantiable even
if they are not bounded. It also provides a
\emph{new-quantifier}~\cite{GP99}, to model name generation and
locality. Nominal logic was introduced at the beginning of this decade
by Gabbay and Pitts~\cite{GP99,Pit01,GP01,Pitts03}.  These first works
have inspired a sequel of papers where bindings and freshness are
introduced in other areas, like nominal algebra~\cite{GM06,GM07,GM09},
equational logic~\cite{CP07}, rewriting~\cite{FG05,FG07},
unification~\cite{UPG03,UPG04}, and Prolog~\cite{CU04,UC05}.

This paper is concerned with \emph{Nominal Unification}, the problem
of deciding if two nominal terms can be made $\alpha$-equivalent by
instantiating their variables by nominal terms. In this instantiation,
variables are allowed to \emph{capture} atoms. Urban, Pitts and
Gabbay~\cite{UPG03,UPG04} describe a sound and complete, but
inefficient (exponential), algorithm for nominal
unification. Fern{\'a}ndez and Gabbay~\cite{FG05} extend this
algorithm to deal with the new-quantifier and locality. Nominal
Logic's equivariance property suggested to Cheney~\cite{CheRTA05} a
stronger form of unification called equivariant unification.  He
proves that equivariant unification and matching are NP-hard problems.
Another variant of nominal unification is permissive unification,
defined by Dowek, Gabbay and Mulligan~\cite{DGM09,DGM10}, that is also
reducible to Higher-Order Pattern Unification.  Calv{\`e}s and
Fern{\'a}ndez describe in~\cite{CF07} a direct but exponential
implementation of a nominal unification algorithm in Maude, and
in~\cite{CF08} a polynomial implementation, based on the use of a
graph representation of terms, and a lazy propagation of
swappings. In~\cite{LV08} we prove that Nominal Unification can be
decided in quadratic time by reduction to Higher-Order Pattern
Unification. The present paper is an extension of this preliminary
paper, where we have simplified the reduction by removing freshness
equations, and we have included the proof of some important properties
of pattern unifiers. In particular, we prove that most general
higher-order pattern unifiers can be written without using other
bound-variable names than the ones used in the presentation of the
unification problem. Moreover, we establish a precise correspondence
between most general nominal unifiers and most general pattern
unifiers. Sections~\ref{sec-removing}, \ref{sec-properties} and
~\ref{sec-correspondence} are completely new in this extended
version. Recently, Calv\`es and Fern\'andez~\cite{Cal10}, and
ourself~\cite{LV10} have independently found direct quadratic nominal
unification algorithms based on the Paterson and Wegman's linear
first-order unification algorithm~\cite{PW78}.

The use of $\alpha$-equivalence and binders in nominal logic
immediately suggests to look at nominal unification from a higher-order
perspective, the one that we adopt in this paper. Some intuitions
about this relation were already roughly described by Urban, Pitts and
Gabbay in~\cite{UPG04}. Cheney~\cite{CheUNIF05} reduces
higher-order pattern unification to nominal unification (here we prove
the opposite reduction).

The main benefit of nominal logic, compared to higher-order logic, is
that it allows the use of binding and $\alpha$-equivalence without the
other difficulties associated with the $\lambda$-calculus. In
particular, with respect to unification, we have that nominal
unification is unitary (most general unifiers are unique) and
decidable~\cite{UPG03,UPG04}, whereas higher-order unification is
undecidable and infinitary~\cite{Luc72,Gol81,Lev98,LV00}.  In this
paper we fully develop the study of nominal unification from the
higher-order logics' view. We show that \emph{full} higher-order
unification is not needed, and \emph{Higher-order Pattern Unification}
suffices to encode Nominal Unification.  This subclass of problems was
introduced by Miller~\cite{Mil91}. Contrarily to general higher-order
unification, higher-order pattern unification is decidable and
unitary~\cite{Mil91,Nip93}. Moreover, unifiability can be decided in
linear time~\cite{Qia96}. All this will lead us to show how to reduce
nominal unification to higher-order pattern unification, and to
conclude its decidability in quadratic deterministic time.

From a higher-order perspective, nominal unification can be seen as a
variant of higher-order unification where: 
\begin{enumerate} 

\item variables are all first-order typed, and constants are of order
at most three, 

\item unification is performed modulo $\alpha$-equivalence, instead of 
the usual $\alpha$ and $\beta$-equivalence, 

\item instances of variables are allowed to capture atoms, contrarily 
to the standard higher-order definition, and

\item apart from the usual equality predicate, we use a \emph{freshness}
predicate $a \fresh t$ with the intended meaning: atom $a$ does not
occur free in $t$.
\end{enumerate}

The third point is the key that makes nominal unification an
interesting subject of research. Variable capture is always a trouble
spot. Roughly speaking, the main idea of this paper is to translate
atoms into bound variables, and variables into free variables with the
list of atoms that they can \emph{capture} as arguments.  The first
point will ensure that, since variables do not have parameters, after
translation, the only arguments of free variables will be list of
pairwise distinct bound variables, hence higher-order patterns.
Moreover, since bound variables will be first-order typed, and
constants third-order typed, the translated problems will be
\emph{second-order} patterns.  The second point is not a
difficulty. Since all nominal variables are first-order typed, their
instantiation does not introduce $\beta$-redexes.  Finally, the fourth
point can also be overcome by translating freshness equations into
equality equations, as described in Section~\ref{sec-removing}.

The remainder of the paper proceeds as follows. After some
preliminaries in Section~\ref{sec-preliminaries}, in
Section~\ref{sec-examples} we illustrate by examples the main ideas of
the reduction at the same time that we show the main features of
nominal unification.  In Section~\ref{sec-removing}, we prove that
freshness equations can be linearly translated into equality
equations.  In Section~\ref{sec-translation}, we show how to translate
a nominal unification problem into a higher-order patterns unification
problem.  Then, after proving some properties of Higher-Order Pattern
Unification in Section~\ref{sec-properties}, we prove that this
translation is effectively a quadratic time reduction, in
Section~\ref{sec-reverse}. In Section~\ref{sec-correspondence}, we
establish a correspondence between nominal unifiers and pattern
unifiers of the translated problems. In particular, we prove that the
translation function and its inverse are monotone w.r.t. the more
general relation, and both translate most general unifiers into most
general unifiers. We conclude in Section~\ref{sec-conclusions}.

\section{Preliminaries}
\label{sec-preliminaries}

In this section we present some basic definitions of Nominal
Unification and Higher-Order Pattern Unification. We will use two
distinct typographic fonts to represent nominal terms and $\lambda$-terms
along this paper.

\subsection{Nominal Unification}

Nominal terms contain \emph{variables} and \emph{atoms}.  Only
variables may be instantiated, and only atoms may be bounded.  They
roughly correspond to the notions of free and bound variables in
$\lambda$-calculus, respectively, but are considered as completely
different entities. However, atoms are not necessarily bounded, and
when they occur free, they are not instantiable.

In \emph{nominal signatures} we have \emph{sorts of atoms} (typically
$\nu$) and \emph{sorts of data} (typically $\delta$) as disjoint sets. 
\emph{Atoms} (typically $\sf a,b,\dots$) have one of the sorts of atoms.
\emph{Variables}, also called \emph{unknowns}, (typically $\sf X,Y,\dots$) 
have a sort of atom or sort of data, i.e. of the form
$\nu\,|\,\delta$. \emph{Nominal function symbols} (typically
$\sf f,g,\dots$) have an \emph{arity} of the form $\tau_1\times\cdots\times\tau_n\to \delta$, where
$\delta$ is a sort of data and $\tau_i$ are sorts given by the grammar
$\tau ::=\nu\,|\,\delta\,|\,\langle\nu\rangle\tau$.
Abstractions have sorts of the form $\langle\nu\rangle\tau$.

\emph{Nominal terms} (typically $\sf t,u,\dots$) are given by the grammar: 
\[
\sf t ::= f(t_1,\dots,t_n)\,|\, a\,|\, a.t \,|\,\pi\mdot X
\] 
where $\sf f$ is a n-ary function symbol, $\sf a$ is an atom, $\pi$ is a permutation
(finite list of swappings), and $\sf X$ is a variable.  They are called
respectively \emph{application}, \emph{atom},
\emph{abstraction} and \emph{suspension}. 
The set of variables of a term $\sf t$ is denoted by $\sf \Vars(t)$.

A \emph{swapping} $\sf (a\,b)$ is a pair of atoms of the same sort. The
effect of a swapping over an atom is defined by $\sf (a\,b)\mdot a = b$
and $\sf (a\,b)\mdot b = a$ and $\sf (a\,b)\mdot c = c$, when $\sf c\neq a,b$.
For the rest of terms the extension is straightforward, in particular,
$\sf (a\,b)\mdot (c.t)=\big((a\,b)\mdot c \big). \big((a\,b)\mdot t\big)$.
A \emph{permutation} is a (possibly empty) sequence of swappings. Its
effect is defined by $\sf (a_1\,b_1)\dots(a_n\,b_n)\mdot t = (a_1\,b_1)
\mdot \big( (a_2\,b_2)\dots(a_n\,b_n)\mdot t\big)$.  Notice that every
permutation $\pi$ naturally defines a bijective function from the set
of atoms to the sets of atoms, that we will also represent as $\pi$.
\emph{Suspensions} are uses of variables with a permutation
of atoms waiting to be applied once the variable is instantiated.
Occurrences of an atom $\sf a$ are said to be \emph{bound} if they are in the
scope of an abstraction of $\sf a$, otherwise are said to be
\emph{free}.

\emph{Substitutions} are finite sets of pairs $\sf[X_1\mapsto
t_1,\dots,X_n\mapsto t_n]$ where $\sf X_i$ and $\sf t_i$ have the same
sort, and the $\sf X_i$'s are pairwise distinct variables.  They can
be extended to sort-respecting functions between terms, and behave
like in first-order logic, hence allowing atom capture. For instance
$\sf [X\mapsto a]a.X = a.a$.  Remember that when applying a
substitution to a suspension, the permutation is immediately applied,
for instance 
$$
\sf [X\mapsto g(a)] f\big( (a\,b)\mdot X, X\big) =
f\big( (a\, b)\mdot g(a), g(a)\big)= f\big(g((a\,b)\mdot a), g(a)\big) = f\big( g(b),g(a)\big)
$$
The domain of a substitution $\sf\sigma=[X_1\mapsto
t_1,\dots,X_n\mapsto t_n]$ is $\sf\dom(\sigma)=\{X_1,\dots,X_n\}$. For
convenience we consider $\sf\dom([X\mapsto X]) = \{X\} \neq \{Y\} =
\dom([Y\mapsto Y])$, although both substitutions have the same effect
when applied to any term.\footnote{We have adopted this definition
  motivated by Remark~\ref{rem-domain}.} Composition of substitutions
is defined by $\sf \sigma_1\circ\sigma_2=[X\mapsto
\sigma_1(\sigma_2(X))\ |\ X\in\dom(\sigma_1)\cup\dom(\sigma_2)]$.  The
restriction of a substitution $\sigma$ to a set of variables $\sf V$,
written $\sf\sigma|_V$, is defined as $\sf\sigma|_V = [X\mapsto
\sigma(X)\ |\ X\in V]$.

A \emph{freshness environment} (typically $\nabla$) is a list of
\emph{freshness constraints} $\sf a\fresh X$ stating that the
instantiation of $\sf X$ cannot contain free occurrences of~$\sf a$.

The notion of \emph{$\alpha$-equivalence} between terms, noted
$\approx$, is defined by means of the following theory:
\[
\begin{array}{c@{\hspace{10mm}}c}
\infrule{}{\sf \nabla \vdash a \approx a}{$\approx$-atom} 
&
\infrule{\sf a \# X\in\nabla \mbox{ for all $\sf a$ such that $\sf \pi\mdot a\neq \pi'\mdot a$}}{\sf \nabla \vdash \pi\,\mdot\,X\approx \pi'\,\mdot\, X}{$\approx$-susp.}
\\[8mm]
\multicolumn{2}{c}{
\infrule{\sf \nabla \vdash t_1\approx t_1' \qquad \cdots\qquad\nabla\vdash t_n\approx t_n'}{\sf \nabla \vdash f(t_1,\dots,t_n)\approx f(t_1',\dots,t_n')}{$\approx$-application}}
\\[8mm]
\infrule{\sf \nabla\vdash t\approx t'}{\sf \nabla \vdash a.t\approx a.t'}{$\approx$-abst-1}
&
\infrule{\sf a\neq a'\quad \nabla\vdash t\approx (a\,a')\mdot t'\quad \nabla\vdash a\# t'}{\sf \nabla\vdash a.t\approx a'.t'}{$\approx$-abst-2}
\end{array}
\]
where the \emph{freshness} predicate $\#$ is defined by:
$$
\begin{array}{c@{\hspace{10mm}}c}
\infrule{\sf a \neq a'}{\sf \nabla \vdash a\# a'}{\#-\mbox{atom}}
&
\infrule{ \sf (\pi^{-1}\mdot a\,\# X)\in\nabla }{\sf \nabla \vdash a\#\pi\,\mdot\,X}{\#-susp.}
\\[8mm]
\multicolumn{2}{c}{
\infrule{\sf \nabla \vdash a\# t_1 \qquad \cdots\qquad\nabla\vdash a\# t_n}{\sf \nabla \vdash a \# f(t_1,\dots,t_n)}{\#-\mbox{application}}}
\\[8mm]
\infrule{}{\sf \nabla \vdash a\# a.t}{\#-\mbox{abst-1}}
&
\infrule{\sf a\neq a'\quad \nabla\vdash a\# t}{\sf \nabla\vdash a\# a'.t}{\#-abst-2}
\end{array}
$$

Their intended meanings are:
\begin{itemize}
\item
$\sf \nabla\vdash a\fresh t$ holds if, for every substitution $\sigma$
respecting the freshness environment $\nabla$ (i.e. avoiding the atom
captures forbidden by $\nabla$), $\sf a$ is not free in $\sf \sigma(t)$;
\item
$\sf \nabla\vdash t\approx u$ holds if, for every substitution $\sigma$
respecting the freshness environment $\nabla$, $\sf t$ and $\sf u$ are
$\alpha$-equivalent.
\end{itemize}

A \emph{nominal unification problem} (typically $\sf P$) is a set of
equations of the form $\sf t\unifn u$, or of the form $\sf a\fresh^?t$,
called \emph{equality equations} and \emph{freshness equations},
respectively.

A \emph{solution} or \emph{unifier} of a nominal problem $\sf P$ is a pair
$\pair\nabla\sigma$ satisfying $\sf \nabla\vdash
a\fresh\sigma(t)$, for all freshness equations $\sf a\fresh^?t\in P$,
and $\sf \nabla\vdash\sigma(t)\approx\sigma(u)$, for all equality
equations $\sf t\unifn u\in P$.  Later, in
Section~\ref{sec-translation}, we will also require solutions to
satisfy $\sf\dom(\sigma)=\Vars(P)$.  In Remark~\ref{rem-domain} we
justify why this does not affect to solvability of nominal problems.

Given two substitutions $\sf\sigma_1$ and $\sf\sigma_2$, and two
freshness environments $\sf\nabla_1$ and $\sf\nabla_2$, we say that
$\sf\nabla_2\vdash \sigma_1(\nabla_1)$, if $\sf \nabla_2\vdash a\fresh
\sigma_1(X)$ holds for each $\sf a\fresh X\in\nabla_1$; and we say
that $\sf\nabla_1\vdash \sigma_1\approx\sigma_2$, if
$\sf\nabla_1\vdash \sigma_1(X)\approx \sigma_2(X)$ holds for all $\sf
X\in \dom(\sigma_1)\cup\dom(\sigma_2)$.  Given a nominal unification
problem $\sf P$, we say that a solution $\pair{\nabla_1}{\sigma_1}$ is
\emph{more general} than another solution $\pair{\nabla_2}{\sigma_2}$,
if  there exists a
substitution $\sigma'$ satisfying $\nabla_2\vdash \sigma'(\nabla_1)$
and
$\sf\nabla_2\vdash\sigma'\circ\sigma_1|_{\dom(\sigma_1)\cup\dom(\sigma_2)}\approx\sigma_2$. As
usual, we say that a solution $\sigma$ is \emph{most general} if, for
any other solution $\sigma'$ more general than $\sigma$, we have also
that $\sigma$ is also more general than $\sf\sigma'$. Most general
nominal unifiers are \emph{unique}, in the usual sense: if
$\sf\sigma_1$ and $\sf\sigma_2$ are both most general, then
$\sf\sigma_1$ is more general than $\sf\sigma_2$, and vice versa.

\begin{example}
  The solutions of the equation $\sf a.X \unifn b.Y$ can not
  instantiate $\sf X$ with terms containing free occurrences of the
  atom $\sf b$, for instance if we apply the substitution $\sf
  [X\mapsto b]$ to both sides of the equation we get $\sf [X\mapsto b]
  (a.X) = a.b$, for the left hand side, and $\sf [X\mapsto b] (b.Y) =
  b.Y$, for the right hand side, and obviously $\sf a.b\unifn b.Y$ is
  unsolvable.

  A most general solution of this equation is $\sf \pair{\{b\# X\}}
  {Y\mapsto (a\,b)\mdot X]}$. Another most general solution is $\sf
  \pair{\{a\# Y\}}{[X\mapsto (a\,b)\mdot Y]}$. Notice that the first
  unifier is equal to the second composed with $\sf \sigma'=[Y\mapsto
  (a\,b)\cdot X]$, hence the second one is more general than the first
  one. Similarly, the first one is more general that the second
  one. Hence, both are equivalent.
\end{example}

\subsection{Higher-Order Pattern Unification}\label{sec-pattern}

In higher-order signatures we have \emph{types} constructed from a set
of \emph{basic types} (typically $\delta,\nu,\dots$) using the grammar
$\tau ::= \delta\,|\,\nu\,|\,\tau\to\tau$, where $\to$ is associative to the
right.  \emph{Variables} (typically $X,Y,Z,x,y,z,a,b,\dots$) and
\emph{constants} (typically $f,c,\dots$) have an assigned type.

\emph{$\lambda$-terms} are built using the grammar 
\[t ::= x\,|\,c\,|\,\lambda
x.t\,|\, t_1\,t_2
\]
where $x$ is a variable and $c$ is a constant, and are typed as usual.
For convenience, terms of the form $(\dots (a\, t_1)\dots t_n)$, where
$a$ is a constant or a variable, will be written as $a(t_1,\dots,t_n)$,
and terms of the form $\lambda x_1.\cdots.\lambda x_n.t$ as $\lambda
x_1,\dots,x_n.t$. We use $\vec{x}$ as a short-hand for
$x_1,\dots,x_n$.  If nothing is said, terms are assumed to be written
in $\eta$-long $\beta$-normal form. Therefore, all terms have the form
$\lambda x_1.\dots.\lambda x_n. h(t_1,\dots,t_m)$, where $m,n\geq 0$,
$h$ is either a constant or a variable, $t_1,\dots,t_m$ have also this form,
and the term $h(t_1,\dots,t_m)$ has a basic type.

Other standard notions of the simply typed $\lambda$-calculus, like
\emph{bound} and \emph{free} occurrences of variables, 
\emph{$\alpha$-conversion, $\beta$-reduction, $\eta$-long 
$\beta$-normal form}, etc. are defined as usual (see~\cite{dowek}). We
will notate free occurrences of variables with capital letters
$X,Y,\dots$, for the sake of readability. The set of free variables of
a term $t$ is denoted by $\Vars(t)$.  When we write an equality
between two $\lambda$-terms, we mean that they are equivalent modulo
$\alpha$, $\beta$ and $\eta$ equivalence. When we write an equality
$=_\alpha$, we mean that they are $\alpha$-equivalent.

\emph{Substitutions} are finite sets of pairs $\sigma=[X_1\mapsto
t_1,\dots,X_n\mapsto t_n]$ where $X_i$ and $t_i$ have the same type
and the $X_i$'s are pairwise distinct variables. They can be extended
to type preserving function from terms to terms as usual. The
\emph{domain} is $\dom(\sigma) =\{X_1,\dots,X_n\}$. We say that a
substitution $\sigma_1$ is \emph{more general} than another
substitution $\sigma_2$, if there exists a substitution $\sigma'$
satisfying $\sigma'\circ\sigma_1(X)=\sigma_2(X)$, for all $X\in
\dom(\sigma_1)\cup\dom(\sigma_2)$.  We say that a variable $X$
\emph{occurs} in a substitution $\sigma$, if $X\in\Vars(\sigma(Y))$,
for some $Y\in\dom(\sigma)$.

A \emph{higher-order unification problem} is a finite set of equations
$P=\{t_1\unif u_1,\dots,t_n\unif u_n\}$, where $t_i$ and $u_i$ have
the same type. A \emph{solution} or \emph{unifier} of a unification
problem $P$ is a substitution $\sigma$ satisfying
$\sigma(t)=\sigma(u)$, for all equations $t\unif u\in P$. We say that a unifier
$\sigma$ is \emph{most general} if, for any other unifier $\sigma'$
more general than $\sigma$, we have $\sigma$ is also more general than
$\sf\sigma'$.

A \emph{higher-order pattern} is a $\lambda$-term where,
when written in $\beta\eta$-normal form, all free variable occurrences
are applied to lists of pairwise distinct bound variables. For
instance, $\lambda x.f(X(x),Y)$, $f(c,\lambda x.x)$ and $\lambda
x,y.X(\lambda z.x(z),y)$ are patterns, while $\lambda
x.f(X(X(x)),Y)$, $f(X(c),c)$ and $\lambda x.\lambda y.X(x,x)$ are
not. Notice that, since $\lambda z.x(z)$ is equivalent to $x$, the
parameters of $X(\lambda z.x(z),y)$ are considered a list of pairwise
distinct bound variables.

\emph{Higher-order pattern unification} is the problem of deciding if
there exists a unifier for a set of equations between higher-order
patterns.  Like in nominal unification, most general pattern unifiers
are unique. Moreover, most general unifiers instantiate variables by
higher-order patterns.

The following is a set of rules defining Nipkow's
algorithm~\cite{Nip93} that computes, when it exists, the most general
unifier of a pattern unification problem.
\[
\arraycolsep 3pt
\begin{array}{rcl}
\lambda x\,.\, s \unif \lambda x\,.\, t& \to  &\langle  \{s\unif t\}, [\ ]\rangle\\[1mm]
a(t_1,\dots, t_n)\unif a(u_1,\dots, u_n) &\to& \langle \{t_1\unif u_1, \dots, t_n\unif u_n\}, [\ ]\rangle\\
& & \mbox{where $a$ is a constant or bound variable}\\[1mm]
Y(\vec{x})\unif a(u_1,\dots,u_m) & \to & \langle \begin{array}[t]{l}\{Y_1(\vec{x})\unif u_1,\dots, Y_m(\vec{x})\unif u_m\} , \\
{}[Y\mapsto \lambda \vec{x}.a(Y_1(\vec{x}),\dots,Y_m(\vec{x})) ]\,\rangle\end{array}\\
& & \mbox{where $Y\not\in \FV(u_1,\dots,u_m)$}\\
& & \mbox{and $a$ is a constant or $a\in\{\vec{x}\}$ }\\[1mm]
X(\vec{x})\unif X(\vec{y}) & \to & \langle \emptyset, [X\mapsto \lambda \vec{x}.Z(\vec{z}) ] \rangle\\
& & \mbox{where $\{\vec{z}\} = \{x_i\, |\, x_i=y_i\}$}\\[1mm]
X(\vec{x})\unif Y(\vec{y}) & \to & \langle \emptyset, [X\mapsto \lambda \vec{x}.Z(\vec{z}), Y\mapsto \lambda \vec{y}\,.\,Z(\vec{z}) ] \rangle\\
& & \mbox{where $X\neq Y$ and $\{\vec{z}\} = \{\vec{x}\}\cap \{\vec{y}\}   $}\\
\end{array}
\]

The rules transform any equation into a pair $\langle$set of
equations, substitution$\rangle$. The algorithm proceeds by replacing
the equation on the left of the rule by the set of equations on the
right. The substitution is applied to the new set of equations, and
used to, step by step, construct the unifier. Therefore, any rule of
the form $t\unif u\to \pair{E}{\rho}$ produces a transformation of the
form
\[
\pair{P\cup\{t\unif u\}}{\sigma} \Rightarrow \pair{\rho(P)\cup E}{\rho\circ\sigma}
\]
The algorithm starts with the pair $\pair{P}{Id}$ and, if $P$ is
solvable, finishes with $\pair{\emptyset}{\sigma}$, where $\sigma$
with domain restricted to $\FV(P)$ is the most general unifier~\cite[Theorem 3.1]{Nip93}.

In the first rule the binder can be
removed because, in Nipkow's presentation, free and bound variable names
are assume to be from distinct sets, and can be distinguished.  The
equations on the right of the second rule may not be normalized,
i.e. the term $\lambda\vec{x}.Y_i(x_1,\dots,x_n)$ may require a
$\eta$-expansion when $u_i$ is not base typed.

There is an algorithm that finds higher-order pattern unifiers, if
exist, in linear time~\cite{Qia96}.

\section{Four Examples}
\label{sec-examples}

In order to describe the reduction of nominal unification to
higher-order pattern unification, we will use the unification problems
proposed in \cite{UPG03,UPG04} as a quiz.

\begin{example}
The nominal equation 
$$\sf 
a.b. f( X_1,b) 
\unifn
b.a. f( a,X_1) 
$$
has no nominal unifiers.  Notice that, although unification is
performed modulo $\alpha$-equivalence, as far as we allow atom
capture, we can not $\alpha$-convert terms before instantiating
them. Therefore, this problem is \emph{not} equivalent to
$$\sf
a.b. f( X_1,b) 
\unifn
a.b. f( b,X_1) 
$$
which is solvable, and must be $\alpha$-converted as
$$\sf
a.b. f( X_1,b) 
\unifn
a.b. f( b,(a\,b)\mdot X_1) 
$$ 

Recall that $\sf (a\,b)\mdot X_1$ means that, after instantiating $\sf
X_1$ with a term that possibly contain $\sf a$ or $\sf b$, we have to
exchange these variables.

According to the ideas exposed in the introduction, we have to replace
every occurrence of $\sf X_1$ by $ X_1(a,b)$, since $\sf \langle a,b\rangle$ is the
list of atoms (bound variables $a,b$) that can be captured. We
get:
$$
\lambda a.\lambda b. f(X_1(a,b),b) 
\unif
\lambda b.\lambda a. f(a,X_1(a,b)) 
$$ 
Since this is a  higher-order unification problem, we
can $\alpha$-convert one of the sides of the equation and get: 
$$
\lambda a.\lambda b. f(X_1(a,b),b)
\unif
\lambda a.\lambda b. f(b,X_1(b,a))
$$
which is unsolvable, like the original nominal equation.
\end{example}

\begin{example}
The nominal equation 
$$
\sf a.b. f(X_2,b)
\unifn
b.a. f(a,X_3)
$$
 is solvable.
Its translation is
$$
\lambda a.\lambda b. f(X_2(a,b),b)
\unif
\lambda b.\lambda a. f(a,X_3(a,b))
$$
The most general unifier of this higher-order pattern unification problem is
$$
\begin{array}{l}
X_2 \mapsto \lambda x.\lambda y. y\\
X_3 \mapsto \lambda x.\lambda y. x
\end{array}
$$

Now, taking into account that the first argument corresponds to the
atom $\sf a$, and the second one to $\sf b$, we can reconstruct the
most general nominal unifier as:
$$
\begin{array}{l}
\sf X_2 \mapsto b\\
\sf X_3 \mapsto a
\end{array}
$$
\end{example}

\begin{example}
  In some cases, there are interrelationships between the instances of
  variables that make reconstruction of unifiers more difficult. This
  is shown with the following example:
\[
\sf a.b. f(b,X_4) 
\unifn
b.a. f(a,X_5)
\]
that is solvable. Its translation results on:
\[
\lambda a.\lambda b. f(b,X_4(a,b))
\unif
\lambda b.\lambda a. f(a,X_5(a,b))
\]
and its most general unifier is:\footnote{The unifier $X_5 \mapsto
  \lambda x.\lambda y. X_4(y,x)$ is equivalent modulo variable
  renaming. In this case we obtain the also equivalent nominal unifier
  $\sf X_5 \mapsto (a\ b)\mdot X_4$.}
\[
X_4\mapsto \lambda x.\lambda y. X_5(y,x)
\]
This higher-order unifier can be used to reconstruct the nominal unifier
\[
\sf X_4 \mapsto (a\ b)\mdot X_5
\]

The swapping $\sf (a\,b)$ comes from the fact that the arguments of
$X_5$ and the lambda abstractions in front have a different order.
\end{example}

\begin{example}\label{ex-four}
  The solution of a nominal unification problem is not just a
  substitution, but a pair $\sf\pair\nabla\sigma$ where
  $\sf \sigma$ is a substitution and $\sf \nabla$ is a freshness
  environment imposing some restrictions on the atoms that can occur
  free in the fresh variables introduced by $\sf \sigma$. The nominal
  equation
$$\sf 
a.b. f(b,X_6)
\unifn
a.a. f(a,X_7) 
$$
has as solution
$$
\begin{array}{l}
\sf \sigma = [X_6 \mapsto (b\ a)\mdot X_7]\\
\sf \nabla = \{b \fresh X_7\}
\end{array}
$$ 
where the freshness environment is not empty and requires instances of
$\sf X_7$ to not contain (free) occurrences of $\sf b$.  Let us see
how this is reflected when we translate the problem into a
higher-order unification problem. The translation of the equation
using the translation algorithm results on:
\begin{equation}
\lambda a.\lambda b. f(b,X_6(a,b))
\unif
\lambda a.\lambda a. f(a,X_7(a,b))
\label{eq1}\end{equation}
After a convenient $\alpha$-conversion we get
$$
\lambda a.\lambda c. f(c,X_6(a,c))
\unif
\lambda a.\lambda c. f(c,X_7(c,b))
$$
The most general unifier is again unique:
$$
\begin{array}{l}
X_6 \mapsto \lambda x.\lambda y. X_8(y,b)\\
X_7 \mapsto \lambda x.\lambda y. X_8(x,y)
\end{array}
$$

Nevertheless, in this case we cannot reconstruct the nominal
unifier. Moreover, by instantiating the free variable $b$, we get
other (non-most general) higher-order unifier without nominal
counterpart. The translation does not work in this case because $b$
occurs free in the right hand side of (\ref{eq1}).  We translate both
atoms and nominal variables as higher-order variables. Occurrences of
nominal variables become free occurrences of variables, and
occurrences of atoms, if are bounded, become bound occurrences of
variables. Therefore, in most cases, after the translation the
distinction atom/variable become a distinction free/bound
variable. However, if atoms are not bounded, as in this case, they are
translated as free variables, hence are instantiable, whereas atoms
are not instantiable.

To avoid this problem, we have to ensure that any occurrence of an
atom is translated as a bound variable occurrence. This is easily
achievable if we add binders in front of both sides of the
equation. Therefore, the correct translation of this problem is:
$$
\lambda a.\lambda b.\lambda a.\lambda b. f(b,X_6(a,b)) 
\unif
\lambda a.\lambda b.\lambda a.\lambda a. f(a,X_7(a,b))
$$ 
where two new binder $\lambda a.\lambda b$ have been introduced in
front of both sides of the equation.  The most general unifier is now:
$$
\begin{array}{l}
X_6 \mapsto \lambda x.\lambda y. X_8(y)\\
X_7 \mapsto \lambda x.\lambda y. X_8(x)
\end{array}
$$
This can be used to reconstruct the nominal substitution:
$$
\begin{array}{l}
\sf X_6 \mapsto (a\ b)\mdot X_8\\
\sf X_7 \mapsto X_8
\end{array}
$$

As far as $X_8(x)$ is translated back as $\sf X_8$, and $X_8(x)$ does
not uses the second argument (the one corresponding to $\sf b$), we
have to add a supplementary condition ensuring that $\sf X_8$ does not
contain free occurrences of $\sf b$. This results on the freshness
environment $\sf \{b \fresh X_8\}$. Then, $X_8(y)$ is translated back
as $\sf (a\,b)\mdot X_8$.
\end{example}

\section{Removing Freshness Equations}
\label{sec-removing}

In this section we show that freshness equations do not contribute to
make nominal unification more expressive. We prove that nominal
unification can be linearly-reduced to nominal unification without
freshness equations. We call this restriction of nominal unification
\emph{equational nominal unification}.  In next sections we will
describe a quadratic reduction of equational nominal unification to
higher-order pattern unification. The absence of freshness equations
makes the reduction to higher-order pattern unification simpler,
compared with the reduction described in the preliminary version of
this paper~\cite{LV08}.

\begin{definition}
  We define the translation of nominal unification problems into
  equational nominal unification problems inductively as follows:
$$
\begin{array}{ll}
\sf\etrans{\{a \fresh^? t\} \cup P} = \{ a.b.t \unifn b.b.t\} \cup \etrans{P} & \mbox{for some $\sf b\neq a$}\\[1mm]
\multicolumn{2}{l}{\sf\etrans{\{t\unifn u\} \cup P} =  \{t\unifn u\} \cup \etrans{P}}\\
\end{array}
$$ 
\end{definition}

\begin{lemma}
  Given a nominal unification problem $\sf P$, its translation into
  equational nominal unification $\sf \etrans{P}$ can be calculated in
  linear time. Hence, $\sf \etrans{P}$ has linear-size on the size of
  $\sf P$.
\end{lemma}

\begin{lemma}
  The pair $\sf \pair\nabla\sigma$ solves $\sf P$, if, and
  only if, $\sf \pair\nabla\sigma$ solves $\sf \etrans{P}$.
\end{lemma}
\begin{proof}
  We first prove that $\sf \langle a\# t,Id\rangle$ is a solution of
  $\sf\{ a.b.t \unifn b.b.t\}$ when $\sf b\neq a$

\[\sf
\infer[\mbox{\footnotesize ($\approx$-abst-2)}]%
{\sf a.b.t \approx b.b.t}
{
    	\infer[\mbox{\footnotesize ($\approx$-abst-2)}]%
	{\sf b.t \approx a.(a\,b)\mdot t}
	{
	\infer*{\sf t \approx t}{}
	&
	\infer*[\mbox{\footnotesize (lemma 2.7)}]{\sf b \# (a\,b)\mdot t}
	{\sf a \# t}
	}
&
   	\infer[\mbox{\footnotesize ($\#$-abst-2)}]%
  	 {\sf a \# b.t}
  	 {\sf a \# t}
}
\]

In this proof we prove $\sf t\approx t$ from an empty set of
assumptions. We can prove that this is always possible, for any term
$\sf t$, by structural induction on $\sf t$.  We also prove $\sf
b\#(a\,b)\mdot t$ from $\sf a\# t$, using Lemma~2.7 of \cite{UPG04}.

Lemma 2.14 of \cite{UPG04} states that $\sf\nabla'\vdash
\sigma(\nabla)$ and $\sf\nabla\vdash t\approx t'$ implies
$\sf\nabla'\vdash\sigma(t)\approx\sigma(t')$.  In particular,
$\sf\nabla\vdash \sigma(a\#t)$ and $\sf a\#t\vdash a.b.t \approx b.b.t$
implies $\sf\nabla\vdash\sigma( a.b.t)\approx\sigma(b.b.t)$.  Therefore,
if $\sf\pair\nabla\sigma$ solves $\sf a\#^? t$, then $\sf
\pair\nabla\sigma$ solves $\sf a.b.t \unifn b.b.t$.

Second, analyzing the previous proof, we see that the inference rules
applied in each situation were the only applicable rules. Therefore,
any solution $\sf\pair\nabla\sigma$ solving $\sf a.b.t
\unifn b.b.t$, also solves $\sf a\#^? t$, because any proof of $\sf
\sf\sigma(a.b.t) \approx \sigma(b.b.t)$ contains a proof of $\sf a\#
\sigma(t)$ as a sub-proof.

From, these two facts we conclude that $\sf a \fresh^? t$ and $\sf
a.b.t \unifn b.b.t$ have the same set of solutions, for any $\sf b\neq
a$. Therefore, $\sf \{a \fresh^? t\} \cup P$ and $\sf \{ a.b.t \unifn
b.b.t\} \cup P$, also have the same set of solutions, for any nominal
unification problem $\sf P$.  From this we conclude that $\sf P$ and
$\sf \etrans{P}$ have the same set of solutions.
\end{proof}

\begin{corollary}
Nominal unification can be linearly-reduced to equational nominal unification.
\end{corollary}

\section{The Translation Algorithm}
\label{sec-translation}

In this section we formalize the translation algorithm. We transform
equational nominal unification problems into higher-order unification
problems. Both kinds of problems are expressed using distinct kinds of
signatures. In nominal unification we have sorts of atoms and sorts of
data. In higher-order this distinction is no longer necessary, and we
will have a \emph{base type} for every sort of atoms $\sf\nu$ or sort
of data $\sf\delta$.  We give a \emph{sort to types translation
  function} that allows us to translate any sort into a type.

\begin{definition}\label{def-translation-types}
The translation function is defined on sorts inductively as follows.
$$
\begin{array}{l}
\trans{\sf\delta} = \delta\\
\trans{\sf\nu} = \nu\\
\trans{\sf\tau_1\times\cdots\times\tau_n\to\tau} = \trans{\sf\tau_1} \to \cdots\to\trans{\sf\tau_n} \to\trans{\tau}\\
\trans{\sf\langle\nu\rangle\tau} = \nu\to\trans{\sf\tau}
\end{array}
$$ 
where $\delta$ and $\nu$ are base types.  
\end{definition}

\begin{remark}
  The translation function for terms depends on \emph{all} the atoms
  occurring in the nominal unification problem. We assume that there
  exists a \emph{fixed}, \emph{finite} and \emph{ordered} list of
  \emph{distinct} atoms $\sf\langle a_1,\dots,a_n\rangle$ used in the
  problem. This seems to contradict the assumption of a countably
  \emph{infinite} set of atoms for every sort. However, this does not
  imply a loss of generality as far as every nominal unification
  problem only contains a finite set of atoms, and its solutions can
  be expressed without adding new atoms (this is a consequence of
  Lemma~\ref{lem-prop-pattern}). Notice also that the nominal
  unification algorithm in~\cite{UPG04} generates unifiers that do not
  introduce new atoms, because, in all transformation rules, the set
  of atoms in the right-hand side are a subset of the set of atoms in
  the left-hand side.

  From now on, we will consider this list given and fixed.

  In~\cite{DGM09,DGM10} they solve this problem using a permission set
  for every variable.  They roughly correspond to the set of atoms
  capturable by this variable.  However, in their case, this set is
  infinite and co-infinite.  In our case, we prove that solutions can
  be expressed using the same finite set of atoms occurring in the
  problem, and the set of capturable atoms of a variable is finite and
  co-finite.
\end{remark}

For every function symbol $\sf f$, we will use a constant with the
same name $f$.  Every atom $\sf a$ is translated as a (bound)
variable, with the same name $a$. For every variable (unknown) $\sf
X$, we will use a (free) variable with the same name $X$.  Trivially,
atom abstractions $\sf a.t$ are translated as lambda abstractions
$\lambda a.t$, and applications $\sf f(t_1,\dots,t_n)$ as applications
$f(t_1,\dots,t_n)$. The translation of suspensions $\sf\pi\mdot X$ is
more complicated, as far as it gets rid of atom capture. Recall that
in all cases we use distinct character fonts for symbols of nominal
logic and symbols of the higher-order framework. The translation is
parametric on a freshness environment. Notice that, although we have
removed freshness equations, nominal unifiers are composed by a
freshness environment and a substitution.

\begin{definition} 
\label{def-translation-terms}
Let $\sf\langle a_1,\dots,a_n\rangle$ be a fixed ordered list of atoms.
The translation function from nominal terms with a freshness
environments $\sf\nabla$ into $\lambda$-terms is defined inductively
as follows.  
$$
\begin{array}{l@{\ \ \ }l}
\transn{a} = a\\
\transn{f(t_1,\dots,t_n)} = f(\transn{t_1},\dots,\transn{t_n})\\
\transn{a.t} = \lambda a.\transn{t}\\
\transn{\pi\mdot X} = X(\transn{\pi\mdot b_1},\dots,\transn{\pi\mdot b_m})
&
\mbox{where $ \sf \langle b_1,\dots,b_m\rangle=\langle a\in\langle a_1,\dots,a_n\rangle\ |\ a\fresh X \not\in \nabla\rangle$}
\end{array}
$$ 
where, for any atom $\sf a:\nu$, $a:\trans{\nu}$ is the corresponding
bound variable, for any function symbol $\sf f:\tau$, $f:\trans{\tau}$
is the corresponding constant, and for any variable $\sf X:\tau$, the
list $\sf \langle b_1,\dots,b_m\rangle$ is the sublist\footnote{Notice
  that we say sublist, not subset, to emphasize that the relative
  order between $\sf a$'s is preserved.} of $\sf \langle
a_1,\dots,a_n\rangle$ composed by the atoms satisfying $\sf a\fresh X
\not\in \nabla$, and
$X:\trans{\nu_1}\to\dots\to\trans{\nu_m}\to\trans{\tau}$ is the
corresponding free variable, where $\sf b_j:\nu_j$.\footnote{Notice
  that $\sf b_j$ and $\sf \pi\mdot b_j$ are of the same sort.}
\end{definition}

\begin{lemma}\label{lem-well-typed-translation}
  For every nominal term $\sf t$ of sort $\sf \tau$, and freshness
  environment $\nabla$, $\trans{t}_\nabla$ is a $\lambda$-term with type
  $\trans{\tau}$.
\end{lemma}

\begin{proof}
  The proof is simple by structural induction on $\sf t$.  The only
  point that needs a more detailed explanation is the case of
  suspensions.  Since $\sf a_i:\nu_i$, $\sf X:\tau$, and
  $X:\trans{\sf\nu_{i_1}}\to\cdots\to
  \trans{\sf\nu_{i_m}}\to\trans{\tau}$, we have $\transn{X} = X
  \left(\transn{a_{i_1}},\dots,\transn{a_{i_m}}\right) :
  \trans{\tau}$.  When $\sf X$ is affected by a swapping $\sf
  (a_{i_j}\,a_{i_k})$ we also have $\transn{(a_{i_j}\,a_{i_k})\mdot X}
  = X \left(\dots, \transn{a_{i_{j-1}}},
    \transn{a_{i_k}},\transn{a_{i_{j+1}}}, \dots,
    \transn{a_{i_{k-1}}}, \transn{a_{i_j}},\transn{a_{i_{k+1}}},
    \dots\right) : \trans{\tau}$ because the suspension is not a valid
  nominal term unless $\sf a_{i_j}$ and $\sf a_{i_k}$ belong to the
  same sort.  The same applies to arbitrary permutations.
\end{proof}

\begin{example}
  Given the nominal term $\sf t=a.b.c.(c\,a)(a\,b)\mdot X$, after
  applying the substitution $\sf \sigma=[X\mapsto f(a, b,c,Y)]$ we get
  $\sf \sigma(t)=a.b.c.f(b,c,a,Y)$.  Let $\sf \langle a,b,c\rangle$ be
  the (ordered) list of atoms of our problem.  The translation of the
  term $\sf t$ w.r.t. $\nabla_1=\emptyset$ results into $\trans{\sf
    t}_{\nabla_1}=\lambda a.\lambda b.\lambda c. X(b, c, a)$ and, the
  translation of the instantiation $\sf\sigma(t)$
  w.r.t. $\sf\nabla_2=\{a\#Y\}$ results into
  $\trans{\sf\sigma(t)}_{\nabla_2}=\lambda a.\lambda b.\lambda c.
  f(b,c,a,Y(c,a))$. There is a $\lambda$-substitution $[X\mapsto
  \lambda a.\lambda b.\lambda c.f(a,b,c,Y(b,c))]$ (described in
  Definition~\ref{def-translation-solutions}) that when applied to
  $\trans{\sf t}_{\nabla_1}$ results into
  $\trans{\sf\sigma(t)}_{\nabla_2}$. Graphically this can be
  represented as the commutation of the following diagram (proved in
  Lemma~\ref{lem-technical2}).

\begin{center}
\begin{diagram}
\sf a.b.c.(c\,a)(a\,b)\mdot X    & 
\rTo^{\sf [X\mapsto f(a,b,c,Y)]}      & 
\sf a.b.c.f(b,c,a,(c\,a)(a\,b)\mdot Y) \\
\dTo_{\trans{\ }_\emptyset} & 
& 
\dTo_{\trans{\ }_{\{a\#Y\}}}     \\
\lambda a.\lambda b.\lambda c. X(b,c,a)  & 
\rTo^{[X\mapsto \lambda a.\lambda b.\lambda c.f(a, b,c,Y(b,c))]}      &
\lambda a.\lambda b.\lambda c.f(b,c,a,Y(c,a))
\end{diagram}
\end{center}
\end{example}

\begin{definition}\label{def-trans-problems}
  Let $\sf\langle a_1,\dots,a_n\rangle$ be an ordered list of
  atoms. The translation function is defined on equational nominal
  problems inductively as follows
$$
\trans{\sf\{t\unifn u\} \cup P} =  \{\lambda a_1.\dots.\lambda a_n.\trans{\sf t}_\emptyset \unif  
                                  \lambda a_1.\dots.\lambda a_n.\trans{\sf u}_\emptyset\} \cup \trans{\sf P}\\
$$
\end{definition}

\begin{lemma}\label{lem-quadratic-pattern}
  Given an equational nominal unification problem $\sf P$, its
  translation $\trans{\sf P}$ is a higher-order pattern unification
  problem.

\noindent
Moreover, the size and the time needed to compute $\trans{\sf P}$ is bounded by the square of the
size of $\sf P$.
\end{lemma}

\begin{proof}
  By Lemma~\ref{lem-well-typed-translation}, $\lambda
  a_1.\dots.\lambda a_n.\trans{\sf t}_\emptyset \unif \lambda
  a_1.\dots.\lambda a_n.\trans{\sf u}_\emptyset$ is an equation
  between $\lambda$-terms of the same type.  Now notice that
  $\transn{\pi\mdot X} = X\left(\transn{\pi\mdot
      b_1},\dots,\transn{\pi\mdot b_m}\right)$ translate the variable
  $\sf X$ into an application of the free variable $X$ to a list of
  pairwise distinct bound variables, because the $\sf b_i$ are all
  different, $\sf\pi$ is a permutation, and we ensure that all atoms
  are translated into bound variables by adding $\lambda$-bindings in
  front of both terms.  Therefore, both sides of the equation are
  higher-order patterns.

  Concerning the size of the translation, we obtain a quadratic bound
  due to the translation of these suspensions.
\end{proof}

Finally, we have to translate solutions of nominal unification
problems into $\lambda$-substitutions.

\begin{definition}
\label{def-translation-solutions}
  Let $\sf \langle a_1,\dots,a_n\rangle$ be a fixed ordered list of
  atoms.  Given a nominal substitution $\sigma$, and a freshness
  environment $\nabla$, we define the following translation
  function
$$
\trans{\sigma}_\nabla = \bigcup_{\sf X\in \dom(\sigma)} \Big[X\mapsto\lambda a_1.\cdots\lambda a_n.\transn{\sigma(X)}\Big] 
$$
\end{definition}

The following remark shows why in some places we require that
solutions $\pair{\nabla}{\sigma}$ of a nominal problem $\sf P$ satisfy
$\sf\dom(\sigma)=\Vars(P)$.

\begin{remark}\label{rem-domain}
  Let $\sf\langle a,b\rangle$ be the fixed list of atoms.  

  Consider the nominal unification problem $\sf
  P_1=\{a.X\unifn b.Y\}$, and its translations as a higher-order pattern unification problem
\[
\trans{\sf P_1}=\trans{\sf \{a.X\unifn b.Y\}} = \{\lambda a.\lambda b.\lambda a.X(a,b) \unif  \lambda a.\lambda b.\lambda b.Y(a,b)\}
\]
The $\lambda$-substitution 
\[
\sigma_1=\trans{\sf[X\mapsto (a\,b)\mdot Y]}_{\sf\{a\#Y\}} =
[X\mapsto \lambda a.\lambda b. Y(a)]
\]
does not solve $\trans{\sf P_1}$. Whereas the $\lambda$-substitution
\[
\sigma_2=\trans{\sf[X\mapsto (a\,b)\mdot Y, Y\mapsto Y]}_{\sf\{a\#Y\}} =
[X\mapsto \lambda a.\lambda b. Y(a) ,\ Y\mapsto \lambda a.\lambda b. Y(b) ]
\]
solves $\trans{\sf P_1}$. Notice that in the first case the domain of
the nominal unifier (as defined in Section~\ref{sec-preliminaries}) is
$\sf \{X\}$, whereas in the other case it is $\sf\{X,Y\}=\Vars(P_1)$.

We will see (Theorem~\ref{the-trans-solution}) that, if
$\sf\Vars(P)\subseteq \dom(\sigma)$ and $\pair\nabla\sigma$ solves
$\sf P$, then $\trans{\sigma}_\nabla$ solves $\sf\trans{P}$. With this
example we see that the first condition in the implication is
necessary.

Now, consider the nominal unification problem $\sf
P_2=\{a.b.(a\,b)X\unifn b.b.(a\,b)X\}$, and its translation as
\[
\trans{\sf P_2}=\{\lambda a.\lambda b.\lambda a.\lambda b. X(b,a) = 
\lambda a.\lambda b.\lambda b.\lambda b. X(b,a)\}
\]
In this case, the pattern substitution $\sigma_1$ is a most general
pattern unifier of $\sf\trans{P_2}$, and $\sigma_2$ is a pattern
unifier, but not a most general one.

Therefore, we have to require $\sf\Vars(P)\supseteq \dom(\sigma)$, if
we want to ensure that the translation not only preserves
unifiability, but also most generality.

Notice that w.l.o.g. we can require most general nominal solutions to
satisfy $\sf\Vars(P)= \dom(\sigma)$, because most general solutions do
not instantiate variables not belonging to $\sf\Vars(P)$, and we can
always add pairs $\sf X\mapsto X$ for all variables occurring in $\sf
P$ and not in $\dom(\sigma)$.

Notice also that in $\sigma_2$ there are two free variables with the
same name $Y$, but distinct types. Be aware that in $Y\mapsto \lambda
a.\lambda b. Y(b)$ the \emph{replaced} $Y$ has two arguments, whereas
the \emph{introduced} $Y$ has only one argument (they have distinct
types). In $\lambda$-calculus this is not a problem.  The reason of
this duplicity is that the translation function is parametric on a
freshness environment $\nabla$.  This is relevant in the case of a
nominal variable. For instance, $\trans{\sf Y}_\emptyset=Y(a,b)$ where
we use the replaced $Y$ with two parameters, and $\trans{\sf
  Y}_{\{a\#Y\}}=Y(b)$ where we use the introduced $Y$ with one
parameter.  If we would like to avoid this duplicity we have to forbid
the use of a variable of the problem in the right-hand side of a
nominal solution. Then, in our example $\sf P_1$, the most general
nominal solution could be written as
$\sf\langle\{a\#Y'\},[X\mapsto(a\,b)Y',Y\mapsto Y']\rangle$.
\end{remark}

To prove that the translation of the solution of a problem is a
solution of the translation of the problem, we start by proving the
following two technical lemmas.

\begin{lemma}\label{lem-technical1}
  For any freshness environment $\sf \nabla$, nominal terms $\sf t$,
  $\sf u$, and atom $\sf a$, we have
\begin{enumerate}
\item $\sf \nabla \vdash a\fresh t$ if, and only if,
$a\not\in \FV(\transn{t})$, and
\item $\sf \nabla \vdash t \approx u$ if, and only if,
$\transn{t} =_\alpha \transn{u}$.
\end{enumerate}
\end{lemma}

\begin{proof}
  The first statement can be proved by routine induction on $\sf t$
  and its translation. Notice that atoms are translated
  \emph{nominally} into variables and that the binding structure is
  also identically translated, hence, the freshness of an atom $\sf a$
  corresponds to the free occurrence of its variable counterpart $a$.
  We here only comment the case $\sf t=\pi\mdot X$, in this case,
  $\transn{\pi\mdot X}= X\left(\transn{\pi\mdot b_1},\dots,
  \transn{\pi\mdot b_m}\right)$, where $\sf b_i\# X\notin\nabla$, for any
  $\sf i\in\{1..m\}$.  Therefore, we can establish the following
  sequence of equivalences $\sf \nabla\vdash a\#\pi\mdot X$ iff $\sf
  \pi^{-1}\mdot a\# X\in\nabla$ iff $\sf \pi^{-1}\mdot
  a\not\in\{b_1,\dots,b_m\}$ iff $\sf a\not\in\{\pi\mdot
  b_1,\dots,\pi\mdot b_m\}$ iff $a\not\in\FV(X(\transn{\pi\mdot
    b_1},\dots, \transn{\pi\mdot b_m}))$ iff $a\not\in\FV(\trans{\sf
    \pi\mdot X})$.

  The proof of the second statement can be done by induction on the
  equivalence $\sf t\approx u$. We only comment the equivalence
  between suspensions: $\sf \pi\mdot X\approx \pi'\mdot X$. Notice
  that, $\sf \pi\mdot X\approx \pi'\mdot X$ if, and only if, for all
  atoms $\sf a$ such that $\sf \pi\mdot a \neq \pi'\mdot a$, we have
  $\sf a\# X\in \nabla$.  This condition is equivalent to: the bound
  variables $\transn{\pi\mdot a}$ and $\transn{\pi'\mdot a}$ are
  passed as a parameter to $X$ in $\transn{\pi\mdot X}$ and
  $\transn{\pi'\mdot X}$ only when $\sf \pi\mdot a = \pi'\mdot
  a$. Finally, this condition is equivalent to $\transn{\pi\mdot X}
  =\transn{\pi'\mdot X}$.
\end{proof}

The first statement of the previous lemma will not be necessary for
our purposes because we have removed freshness equations.

\begin{lemma}\label{lem-technical2}
  For any freshness environment $\sf \nabla$, nominal substitution
  $\sf \sigma$, and nominal term $\sf t$ satisfying
  $\sf\Vars(t)\subseteq\dom(\sigma)$, we have $\trans{\sigma}_\nabla(\trans{\sf t}_\emptyset) =
  \transn{\sigma(t)}$.
\end{lemma}

\begin{proof}
  Again this lemma can be proved by structural induction on $\sf
  t$. We only sketch the suspension case.  Let $\sf t=\pi\mdot X$. We
  have the equalities:
$$
\begin{array}{lll}
\trans{\sigma}_\nabla(\trans{\sf \pi\mdot X}_\emptyset)&=&
 [\dots,X\mapsto \lambda a_1\dots\lambda a_n\,.\,\transn{\sigma(X)},\dots]\
( X (\transn{\pi\mdot a_1}, \dots, \transn{\pi\mdot a_n}))\\
&=& (\lambda a_1\dots\lambda a_n.\transn{\sigma(X)})\
\left(\transn{\pi\mdot a_1}, \dots, \transn{\pi\mdot a_n}\right)\\
&=&[a_1\mapsto \transn{\pi\mdot a_1},\dots,a_n\mapsto \transn{\pi\mdot a_n}]\ \left(\transn{\sigma(X)}\right)\\
&=&\transn{\pi\mdot\sigma(X)}\\
&=&\transn{\sigma(\pi\mdot X)}
\end{array}
$$

Notice that in the first equality we use $\sf
X\in\Vars(t)\subseteq\dom(\sigma)$, hence $X\in \dom(\trans{\sigma}_\nabla)$.
\end{proof}

\begin{example}
  Let be $\sf t=f((a\,b)\mdot X,(a\,b)\mdot Y)$, $\sf\nabla=\{b\#Y\}$
  and \hbox{$\sf\sigma=[ X\mapsto b.a,$}\ \hbox{$\sf Y\mapsto Y]$}.
  We will have
\[
\arraycolsep 0mm
\begin{array}{ll}
\trans{\sigma}_\nabla&=
\trans{\sf[X\mapsto b.a,Y\mapsto Y]}_{\sf\{b\#Y\}}
=[X\mapsto \lambda a.\lambda b.\trans{\sf b.a}_{\sf\{b\#Y\}},\ Y\mapsto\lambda a.\lambda b.\trans{\sf Y}_{\sf\{b\#Y\}}]\\
&=[X\mapsto \lambda a.\lambda b.\lambda b.a,\ Y\mapsto\lambda a.\lambda b.Y(a)]\\[2mm]
\trans{\sf t}_{\emptyset}&=
\trans{\sf f\big((a\,b)\mdot X,(a\,b)\mdot Y\big)}_{\emptyset} = f\big(X (b,a),Y(b,a)\big)\\[2mm]
\trans{\sigma(\sf t)}_\nabla&=\trans{\sf[X\mapsto b.a,Y\mapsto Y]\,f\big((a\,b)\mdot X,(a\,b)\mdot Y\big)}_{\sf\{b\#Y\}} = 
\trans{\sf f(a.b,(a\,b)\mdot Y)}_{\sf\{b\#Y\}}\\
&=f(\lambda a.b,Y(b))
\end{array}
\]
Now, we have 
$$
\trans{\sigma}_\nabla \left(\trans{\sf t}_{\emptyset}\right) \begin{array}[t]{l}
= f\big(\underline{(\lambda a.\lambda b.\lambda b.a)(b,a)},(\lambda a.\lambda b.Y(a))(b,a)\big) =
f(\lambda c.b,Y(b)) = f(\lambda a.b,Y(b))\\[1mm]=\trans{\sigma(\sf t)}_\nabla\end{array}
$$

Notice that the substitution resulting form the $\beta$-reduction 
of the underlined redex
needs to avoid a capture of $b$. This is done replacing the bound
variable $b$ by $c$. In the following section we will see that, in
pattern unification, we can do this without using new bound variable
names. In this case, we could have used $a$ instead of $c$.
\end{example}

From these two lemmas we can prove the following results.

\begin{theorem}\label{the-trans-solution}
  For any freshness environment $\sf\nabla$, equational nominal
  unification problem $\sf P$, and nominal substitution $\sf\sigma$
  with $\sf \Vars(P)\subseteq\dom(\sigma)$ , we have that
  $\sf\pair\nabla\sigma$ solves the equational nominal
  unification problem $\sf P$, if, and only if,
  $\trans{\sigma}_\nabla$ solves the pattern unification problem
  $\trans{\sf P}$.
\end{theorem}

\begin{proof}
  By definition of nominal solution, the pair $\sf
  \pair\nabla\sigma$ solves $\sf P$ iff
$$
\sf \nabla \vdash \sigma(t) \approx \sigma(u) \hspace{5mm}\mbox{for all $\sf t\unifn u\in P$}
$$
By Lemma~\ref{lem-technical1} this is equivalent to:
$$
\transn{\sigma(t)} =_\alpha \transn{\sigma(u)} \hspace{5mm}\mbox{for all $\sf t\unifn u\in P$}
$$
and, by Lemma~\ref{lem-technical2} this is equivalent to:
$$
\trans{\sigma}_\nabla(\trans{\sf t}_\emptyset) =
\trans{\sigma}_\nabla(\trans{\sf u}_\emptyset) \hspace{5mm}\mbox{for all $\sf t\unifn u\in P$}
$$
Since the substitution $\trans{\sigma}_\nabla$ does not instantiate the variables $a_1,\dots,a_n$, this is equivalent to (see Remark~\ref{rem-lambda}):
$$
\trans{\sigma}_\nabla\Big(\lambda a_1.\dots.\lambda
a_n. \trans{\sf t}_\emptyset\Big) =
\trans{\sigma}_\nabla \Big(\lambda a_1.\dots.\lambda
a_n.\trans{\sf u}_\emptyset\Big) \hspace{5mm}\mbox{for all $\sf t\unifn u\in P$}
$$
where $\sf \langle a_1,\dots,a_n\rangle$ is the list of atoms occurring in $\sf P$.

Finally, since $\trans{\sf t\unifn u}= \lambda a_1.\dots.\lambda a_n. \trans{\sf t}_\emptyset 
\unif \lambda a_1.\dots.\lambda
a_n. \trans{\sf u}_\emptyset$, this is equivalent to
$\trans{\sigma}_\nabla$ solves
$\trans{\sf P}$.
\end{proof}

The proof of Theorem~\ref{the-trans-solution} also allows us to prove
that $\sf\pair\nabla\sigma$ solves $\sf t\unifn u$, if, and
only if, $\trans{\sigma}_\nabla$
solves $\trans{\sf t}_\emptyset\unif\trans{\sf
  u}_\emptyset$. Therefore, it seems unnecessary to add the
$\lambda$-bindings $\lambda a_1.\cdots.\lambda a_n$ in front of both
sides of the higher-order equations, as was suggested in
Example~\ref{ex-four}. The following remark illustrates what would
happen if we had defined translation of equations in this way.

\begin{remark}\label{rem-lambda}
  Assume that we had defined $\trans{\sf t\unifn u}=\trans{\sf
    t}_\emptyset \unif \trans{\sf u}_\emptyset$, instead of the
  definition we have for $\trans{\sf t\unifn u}$ with the external
  lambda's.

  The translation of the unsolvable nominal equation $\sf a\unifn b$
  would result into $a\unif b$ which is solvable by $[a\mapsto b]$
  (notice that, in this case, atoms are translated into free
  variables). The example does not contradict
  Theorem~\ref{the-trans-solution} because the substitution $[a\mapsto
  b]$ is not the translation of any nominal substitution, i.e.  there
  does not exists a freshness environment $\nabla$ and a nominal
  substitution $\sigma$ such that $\trans{\sigma}_\nabla =[a\mapsto b]$.  If we introduce
  the external $\lambda$-bindings we get the unsolvable higher-order
  unification problem $\lambda a.\lambda b.a\unif \lambda a.\lambda
  b.b$.

  On the other hand, the translation of the solvable nominal equation
  of Example~\ref{ex-four} would be
$$
\trans{\sf P}=\trans{\sf \{a.b.f(b,X_6)\unifn
  a.a.f(a,X_7)\}}= \{\lambda a.\lambda b.f(b,X_6(a,b))\unif \lambda a.\lambda a.f(a,X_7(a,b))\}
$$
that is not a higher-order pattern unification problem (notice that
Lemma~\ref{lem-quadratic-pattern} does not hold if we do not introduce
the external $\lambda$-bindings).

The translation of its nominal most general solution is
$$
\trans{\sigma}_\nabla = \trans{\sf[X_6
\mapsto (b\ a)\mdot X_7]}_{\sf\{b \fresh X_7\}}=[X_6 \mapsto \lambda a.\lambda
b.X_7(b), X_7 \mapsto \lambda a.\lambda b.X_7(a)] 
$$ 
In this case, $\trans{\sigma}_\nabla$ is a higher-order
unifier of $\trans{\sf P}$, as Theorem~\ref{the-trans-solution}
predicts. However, it is not a most general unifier, and we are
interested in translating most general solutions into most general
solutions.
\end{remark}

\begin{theorem}\label{the-one-direction}
  If the equational nominal unification problem $\sf P$ is solvable,
  then the higher-order pattern unification problem $\trans{\sf P}$ is
  also solvable.
\end{theorem}

\begin{proof}
  The theorem is a direct consequence of
  Theorem~\ref{the-trans-solution}.
\end{proof}

The opposite implication of Theorem~\ref{the-one-direction} can not be
directly proved from Theorem~\ref{the-trans-solution}, because
$\trans{\sf P}$ should have solutions that are not of the form
$\trans{\sigma}_\nabla$, for any solution $\pair\nabla\sigma$ of $\sf
P$.

\section{Some Properties of Pattern Unification}
\label{sec-properties}

In this section we prove some fundamental properties of Higher-Order
Pattern Unification. In particular, we prove that we can express most
general unifiers of pattern unification problems only using
bound-variable names and types already used in the problem.  This
property is used in next sections in the translation of pattern
unifiers into nominal unifiers.

In the following example we note that in the solution of pattern
unification problems it is important to \emph{save names} of bound
variables. In the following we will distinguish between variables and
variable names.  For instance $\lambda x.\lambda x.x$ has three
occurrences of variables, two distinct variables, with one unique
variable name. Notice that $\alpha$-conversion preserves the number of
variables, but may change the number of names.

\begin{example}
  Consider the nominal problem $\sf a.X \unifn a.f(b.Y)$.  Its
  translation is $\lambda a.\lambda b.\lambda a.X(a,b)\unif \lambda
  a.\lambda b.\lambda a. f(\lambda b.Y(a,b))$.  An $\alpha$-conversion
  results in $\lambda a.\lambda b.\lambda c.X(c,b)\unif \lambda
  a.\lambda b.\lambda c. f(\lambda d.Y(c,d))$ and it shows that the
  parameters of $X$ and $Y$ are in fact different.  A most general
  solution is $[X \mapsto \lambda c.\lambda b. f(\lambda
  d.Y(c,d))]$. Since $\sf Y$ is translated as $Y(a,b)$, we would have
  to translate back $Y(c,d)$ as $\sf (a\,c)(d\,b)\mdot Y$. And, since
  substitutions like $\sf [X\mapsto t]$ are translated as $[X\mapsto
  \lambda a.\lambda b.  \transn{t}]$, we would have to translate back
  $[X\mapsto \lambda c.\lambda b.  \transn{t}]$ as $[\sf X\mapsto
  (a\,c)\mdot t]$.  Therefore, our pattern unifier had to be
  translated back as $\sf [X\mapsto (a\, c)\mdot f(d.(a\,c)(d\,
  b)\mdot Y)]$.  However, the list of atoms is fixed as the list of
  atoms occurring in the problem, hence, we know how to translate $\sf
  a$ and $\sf b$ as $a$ and $b$ and vice versa, but we do not know how
  to translate back $c$ and $d$. Here it is done introducing new
  atoms. However, the use of an infinite list of atom names would
  imply that the list of arguments of a variable (the list of
  capturable atoms) would be infinite.
\end{example}

If we look at Nipkow's transformation rules described in
Subsection~\ref{sec-pattern}, it seems that no new bound-variable
names are introduced. However, this is not true. There are three
places where their introduction is hidden. In the following we
illustrate these cases.
\begin{enumerate}
\item It is assumed that equations have the same most external
  $\lambda$-bindings, i.e. that they are of the form
  $\lambda\vec{x}.s\unif\lambda\vec{x}.t$. If this is not the case, we
  have to $\alpha$-convert one of the sides. However, this is not
  always possible without introducing new bound-variable names.  For
  instance, if we have the equation $\lambda x.\lambda y.\lambda
  y.X(x,y) \unif \lambda y.\lambda y.\lambda x.Y(x,y)$, after
  $\alpha$-converting the two most external $\lambda$-binder, we get
  $\lambda x.\lambda y.\lambda y.X(x,y) \unif \lambda x.\lambda
  y.\lambda x.Y(x,y)$, that needs a new bound-variable name to obtain
  the same $\lambda$-binders in both sides, by means of
  $\alpha$-conversion.  Using a \emph{new name} $z$ we would get
  $\lambda x.\lambda y.\lambda z.X(x,z) \unif \lambda x.\lambda
  y.\lambda z.Y(z,y)$.

\item In the flex-rigid rule the terms $u_i$ may not be of first-order
  type. In this case, we need to $\eta$-expand some subterms.  For
  instance, the rule transforms $\lambda x.X(x) \unif \lambda
  x.f(\lambda x.g(x))$ into the equation $\lambda x.X_1(x) \unif
  \lambda x.\lambda x.g(x)$ and the substitution $\big[X\mapsto \lambda
  x.f(X_1(x))\big]$. The left-hand side of the equation needs to be
  $\eta$-expanded, and we can not use the name $x$. Using a \emph{new
    name} $z$, and $\alpha$-converting we would get $\lambda x.\lambda
  z.X_1(x,z) \unif \lambda x.\lambda z.g(z)$.

\item When we compute a substitution for a variable, it must be
  applied to all the occurrences of the variable, and this may
  involve a $\beta$-reduction. Some $\beta$-reductions need to
  introduce new names to avoid variable-captures. For instance, if we
  have the equations $\big\{\lambda x.\lambda y. X(x,y) \unif \lambda
  x.\lambda y.f(\lambda x.Y(x,y)),\ \lambda x.\lambda y.Z(x,y) \unif
  \lambda x.\lambda y.X(y,x)\big\}$, after solving the first one we get
  $\big[X\mapsto \lambda x.\lambda y.f(\lambda x.Y(x,y))\big]$ that must be
  substituted in the second equation. We get, $\lambda x.\lambda
  y.Z(x,y) \unif \lambda x.\lambda y.\big(\lambda x.\lambda y.f(\lambda
  x.Y(x,y))\big)(y,x)$. The $\beta$-reduction using the standard
  substitution algorithm introduces a \emph{new name} $z$ to avoid the
  capture of the variable $x$, giving $\lambda x.\lambda y.Z(x,y)
  \unif \lambda x.\lambda y.f(\lambda z.Y(z,x))$
\end{enumerate}

In the following we show how we can overcome these problems. One of
the ideas is using a kind of swapping for $\lambda$-calculus, instead of the
usual substitution, like it is done in nominal terms.

\begin{definition}
Given two variables $x,y$, and a $\lambda$-term $t$, we
define the swapping of $x$ and $y$ in $t$, noted by $(x\,y)\mdot t$
inductively as follows 
$$
\begin{array}{l}
(x\,y)\mdot x = y\\
(x\,y)\mdot y = x\\
(x\,y)\mdot z = z \hspace{10mm}\mbox{if $z\neq x,y$}\\
(x\,y)\mdot c = c\\
(x\,y)\mdot \big(\lambda z.t\big) = \lambda \big((x\,y)\mdot z\big).\big((x\,y)\mdot t\big)\\
(x\,y)\mdot \big(a(t_1,\dots,t_n)\big) = \big((x\,y)\mdot a\big)\big((x\,y)\mdot t_1,\dots, (x\,y)\mdot t_n\big)
\end{array}
$$
where $c$ is a constant and $a$ is a constant or a variable.
\end{definition}

Notice that this swapping is distinct from the swapping on nominal
terms. In particular $(a\,b)X = X$, and we do not keep suspensions.
In some cases its application results into an $\alpha$-equivalent
term, but in general the result is a different term.

\begin{remark}
  In $\lambda$-calculus, following the Barendregt variable convention,
  operations are defined on classes of $\alpha$-equivalent terms,
  rather than on particular terms.  This, for instance, allows us to
  freely $\alpha$-convert terms in substitutions in order to avoid
  variable capture. Therefore, (although it is often omitted) we have
  to prove that the operation is independent of the representative of
  the class that we take.  The previous swapping operation is defined
  for particular terms. However, the following lemma ensures that it
  can be extended to $\alpha$-equivalent classes of terms.  Barendregt
  variable convention suggests to use distinct variable names for
  distinct variables. Here, since we try to avoid the introduction of
  new variable names, we do not use the convention, and work with
  particular terms.
\end{remark}

\begin{lemma}\label{lem-swap}
  For any term $t$ and variables $x$ and $y$, we have 
  $$
  (x\, y)\mdot t=_\alpha [x\mapsto y,y\mapsto x]t
  $$
  where $[x\mapsto y,y\mapsto x]$ changes $x$ by $y$ and $y$ by $x$ in
  $t$, simultaneously.

  In particular, if $x,y\not\in\FV(t)$, then $(x\, y)\mdot t=_\alpha t$.
\end{lemma}

\begin{proof}
  By structural induction on $t$. For one of the cases of
  $\lambda$-abstraction, for instance, we have 
  $$
  \begin{array}{lll}
    (x\,y)\lambda x.t 
    & = \lambda y.(x\,y)t & \mbox{By induct. hypothesis}\\
    & = \lambda y.[x\mapsto y,y\mapsto x] t & \mbox{Let be $z\not\in\FV(t)\cup\{x,y\}$} \\
    & = \lambda y.[z\mapsto y][y\mapsto x][x\mapsto z]t & \mbox{Since $y\not\in\FV([y\mapsto x][x\mapsto z]t)$}\\
    & =_\alpha \lambda z.[y\mapsto x][x\mapsto z]t & \mbox{Since $z\neq x,y$}\\
    & = [y\mapsto x] \lambda z.[x\mapsto z]t & \mbox{Since $z\not\in\FV(t)$}\\
    & =_\alpha  [y\mapsto x]\lambda x.t & \mbox{Since $x\not\in\FV(\lambda x.t)$} \\
    & = [x\mapsto y,y\mapsto x]\lambda x.t
  \end{array}
  $$
\end{proof}

\begin{lemma}\label{lem-beta}
  If $\vec{y}$ is a list of pairwise distinct variable
  names\footnote{Notice that we do not require $\vec{x}$ to be
    pairwise distinct. If they are also pairwise distinct, then
    $\Pi_n(\vec{x},\vec{y})=(x_n\,y_n)\dots (x_1\,y_1)$.}  ,
  $|\vec{y}|=|\vec{x}|=n$ and $\{\vec{y}\}\cap \FV(\lambda \vec{x}.t)
  = \emptyset$, then
$$
(\lambda \vec{x}.t)(\vec{y}) =
\Pi_n(\vec{x},\vec{y}) \mdot t
$$ 
where $\Pi_n(\vec{x},\vec{y})$ is a permutation on the names
$\vec{x},\vec{y}$ defined inductively as 
$$
\begin{array}{l}
\Pi_1(\langle
x\rangle,\langle y\rangle) = (x\,y)\\
\Pi_n\big(\langle
x_1,\dots,x_n\rangle,\langle y_1,\dots,y_n\rangle\big) =
\Pi_{n-1}\big(\big\langle(x_1\,y_1)\mdot x_2,\dots,(x_1\,y_1)\mdot
x_n\big\rangle ,\big\langle y_2,\dots,y_n\big\rangle\big)\mdot
(x_1\,y_1)
\end{array}
$$
\end{lemma}

\begin{proof}
  By induction on the length $n$ of both vectors.  Obviously, the
  variable $x_1$ is not free in $\lambda x_1.\lambda x_2,\dots,
  x_n.t$.  By assumption, the variable $y_1$ is neither free in this
  term. 

  From $\FV(\lambda x_2,\dots,x_n.t) \subseteq \FV(\lambda
  \vec{x}.t)\cup\{x_1\}$, and $x_1,y_1\not\in\FV(\lambda \vec{x}.t)$,
  we have $\FV((x_1\,y_1)\mdot(\lambda x_2,\dots,x_n.t)) \subseteq
  \FV(\lambda \vec{x}.t)\cup\{y_1\}$. Since
  $y_1\not\in\{y_2,\dots,y_n\}$ and $\{\vec{y}\}\cap \FV(\lambda
  \vec{x}.t) = \emptyset$, we have $\{y_2,\dots,y_n\}\cap
  \FV((x_1\,y_1)\mdot(\lambda x_2,\dots,x_n.t)) = \emptyset$.
  Therefore, we can apply the induction hypothesis to the term
  $(x_1\,y_1)\mdot(\lambda x_2,\dots,x_n.t)$ and the vector
  $(y_2,\dots,y_n)$, obtaining
$$
\begin{array}{lll}
(\lambda \vec{x}.t)(\vec{y}) 
  & =_\alpha (\lambda y_1.(x_1\,y_1)\mdot(\lambda x_2,\dots,x_n.t)) (y_1,y_2,\dots,y_n)
  &\mbox{By  Lemma~\ref{lem-swap}}\\[1mm]
  & =_\beta ((x_1\,y_1)\mdot(\lambda x_2,\dots,x_n.t))(y_2,\dots,y_n)
  &\mbox{By $\beta$-reduction}\\[1mm]
  & = (\lambda (x_1\,y_1)\mdot x_2,\dots,(x_1\,y_1)\mdot x_n.(x_1\,y_1)\mdot t)(y_2,\dots,y_n)
  &\mbox{By def. of swapping}\\[1mm]
  & = \Pi_{n-1}\big(\langle(x_1\,y_1) \mdot x_2,\dots,(x_1\,y_1) \mdot
  x_n\rangle,\langle y_2,\dots,y_n\rangle\big) \mdot (x_1\,y_1) \mdot t
  & \mbox{By induct. hypothesis}\\[1mm]
  & = \Pi_n(\vec{x},\vec{y}) \mdot t
\end{array}
$$
\end{proof}

Now we will describe a variant of the higher-order pattern unification
algorithm of Section~\ref{sec-pattern}. In this variant, external
$\lambda$-binders are $\alpha$-converted explicitly and the flex-rigid
rule has been replaced by a new rule where $\eta$-expansion is made
explicit, i.e.  the terms $u_i$ are base-typed, thus the right-hand
side does not need to be $\eta$-expanded, like in the original
rule. Moreover, $\beta$-redexes are removed using swappings, according
to Lemma~\ref{lem-beta}, since we are dealing with patterns.

\begin{definition}
\label{def-new-rule}
We assume unoriented equations and define the following set of transformation rules
over higher-order pattern equations:
$$ 
\arraycolsep 2pt
\begin{array}{rcl}
\multicolumn{3}{l}{\mbox{\bf $\alpha$-transformation:}}\\
\lambda \vec{w}.\lambda x.t \unif \lambda \vec{w}.\lambda y.u &\to&
\big\langle\lambda \vec{w}.\lambda x.t \unif \lambda \vec{w}.(x\ y)\mdot(\lambda y.u), [\ ]\big\rangle\\[1mm]
&&\mbox{\rm if $x\not\in\FV(u)$}\\[2mm]
\lambda \vec{w}.\lambda x.t \unif \lambda \vec{w}.\lambda x.u &\to&
\big\langle\lambda \vec{w}.t \unif \lambda \vec{w}.u, [\ ]\big\rangle\\[1mm]
&&\mbox{\rm if $x\not\in\FV(t)$ and $x\not\in\FV(u)$}\\[2mm]
\lambda \vec{w}.\lambda x.t \unif \lambda \vec{w}.\lambda x.u &\to&
\big\langle\lambda \vec{w}.\lambda x.t \unif \lambda \vec{w}.\lambda x.u,[X\mapsto\lambda \vec{y}.Z(\vec{z})]\big\rangle\\[1mm]
&&\mbox{\rm if $x\not\in\FV(t)$, $X(\vec{y})$ is a subterm of $u$,}\\
&&\mbox{ $x\in\{\vec{y}\}$ and $\{\vec{z}\}=\{\vec{y}\}\setminus\{x\}$}\\[2mm]
\multicolumn{3}{l}{\mbox{\bf Rigid-rigid:}}\\
\lambda\vec{w}.a(t_1,\dots, t_n)\unif \lambda\vec{w}.a(u_1,\dots, u_n) &\to& \big\langle \{\lambda\vec{w}.t_1\unif \lambda\vec{w}.u_1, \dots, \lambda\vec{w}.t_n\unif \lambda\vec{w}.u_n\}, [\ ]\big\rangle\\[2mm]
\multicolumn{3}{l}{\mbox{\bf Flex-rigid:}}\\
\lambda\vec{w}.X(\vec{x})\unif
\lambda\vec{w}.a(\lambda\vec{y_1}.u_1,\dots,\lambda\vec{y_m}.u_m) &\to& 
\begin{array}[t]{l}\Big\langle
\big\{\begin{array}[t]{c}
\lambda\vec{w}.\lambda\vec{y_1}.X_1(\vec{z_1})\unif \lambda\vec{w}.\lambda\vec{y_1}.u_1\ ,\\
\dots\\ 
\lambda\vec{w}.\lambda\vec{y_m}.X_m(\vec{z_m})\unif \lambda\vec{w}.\lambda\vec{y_m}.u_m
\big\} ,\end{array}\\
\mbox{}[X\mapsto \lambda \vec{x}.a(\lambda\vec{y_1}.X_1(\vec{z_1}),\dots,\lambda\vec{y_m}.X_m(\vec{z_m})) ]\Big\rangle
\end{array}\\[1mm]
&&\begin{array}[t]{l}
\mbox{\rm if $X\not\in \FV(u_i)$, $a$ is a constant or $a\in\{\vec{x}\}$,} \\
\mbox{and $\{\vec{z_i}\}=\{\vec{x}\}\cup\{\vec{y_i}\}$, for $i=1,\dots,m$.}
\end{array}\\[2mm]
\multicolumn{3}{l}{\mbox{\bf Flex-flex:}}\\
\lambda\vec{w}.X(\vec{x})\unif \lambda\vec{w}.X(\vec{y}) & \to&  
\big\langle \emptyset, [X\mapsto \lambda \vec{x}.Z(\vec{z}) ] \big\rangle\\[1mm]
&& \mbox{\rm where $\{\vec{z}\} = \{x_i\, |\, x_i=y_i\}$}\\[2mm]
\lambda\vec{w}.X(\vec{x})\unif \lambda\vec{w}.Y(\vec{y}) & \to&  
\big\langle \emptyset, [X\mapsto \lambda \vec{x}.Z(\vec{z}), Y\mapsto \lambda \vec{y}\,.\,Z(\vec{z}) ] \big\rangle\\[1mm]
&& \mbox{\rm where $X\neq Y$ and $\{\vec{z}\} = \{\vec{x}\}\cap \{\vec{y}\}   $}\\
\end{array}
$$

These transformations are applied as follows. The equation on the
left-hand side is replaced by the equations in the first component of
the right-hand side, and then the substitution in the second component
of the right-hand side is applied to all the equations. If this
substitution introduces $\beta$-redexes, they are removed using
swappings, according to Lemma~\ref{lem-beta}. Moreover, all the
substitutions are composed to compute the resulting unifier.  In other
words, the transformation is applied as follows $\langle \{e\}\cup
E,\sigma \rangle \to \langle \sigma'(E'\cup E)\downarrow_\beta,
\sigma'\circ\sigma \rangle$, if we have a transformation $e \to \langle
E',\sigma' \rangle$.
\end{definition}

With the following examples, we illustrate how these rules solve
the problems concerning the introduction of new bound variable names
described previously, at the beginning of this section.

\begin{example}
Given the equation $\lambda x.\lambda y.\lambda y.X(x,y) \unif \lambda y.\lambda y.\lambda x.Y(x,y)$
the application of the first $\alpha$-transformation rule gives us
$\lambda x.\lambda y.\lambda y.X(x,y) \unif \lambda x.\lambda x.\lambda y.Y(y,x)$.
A second application of this $\alpha$-transformation gives us
$\lambda x.\lambda y.\lambda y.X(x,y) \unif \lambda x.\lambda y.\lambda x.Y(x,y)$.
Now, the first $\alpha$-transformation rule is no longer applicable.
However, we can apply the third $\alpha$-transformation rule, that
instantiates $[X\mapsto \lambda x.\lambda y.X'(y)]$, and gives the
equation $\lambda x.\lambda y.\lambda y.X'(y) \unif \lambda x.\lambda
y.\lambda x.Y(x,y)$.  Now, applying the second $\alpha$-transformation
rule, we obtain $\lambda y.\lambda y.X'(y) \unif \lambda y.\lambda
x.Y(x,y)$.
Again, we can apply the
third $\alpha$-transformation rule, that instantiates $[Y\mapsto \lambda x.\lambda y.Y'(x)]$, and gives
$\lambda y.\lambda y.X'(y) \unif \lambda y.\lambda x.Y'(x)$.
The first  $\alpha$-transformation rule gives
$\lambda y.\lambda y.X'(y) \unif \lambda y.\lambda y.Y'(y)$.
Finally, the second  $\alpha$-transformation rule gives
$\lambda y.X'(y) \unif \lambda y.Y'(y)$.

This last equation can be solved applying the second flex-flex rule.
The resulting unifier is 
$$
\begin{array}{l}
\big[X'\mapsto\lambda y.Z(y),\ Y'\mapsto\lambda y.Z(y)\big]\circ
\big[Y\mapsto \lambda x.\lambda y.Y'(x)\big]\circ
\big[X\mapsto \lambda x.\lambda y.X'(y)\big]\Big|_{\{X,Y\}}\\[1mm]=
\big[X\mapsto \lambda x.\lambda y.Z(y),\, Y\mapsto \lambda x.\lambda y.Z(x)\big]
\end{array}
$$
\end{example}

\begin{example}
  The new flex-rigid rule transforms $\lambda x.X(x) \unif \lambda
  x.f(\lambda y.a)$ into the equation $\lambda x.\lambda y.X_1(x,y)
  \unif \lambda x.\lambda y.a$ and the substitution $[X\mapsto \lambda
  x.f(\lambda y.X_1(x,y))]$.  The original flex-rigid rule would give
  us $\lambda x.X_1(x) \unif \lambda x.\lambda y.a$, that conveniently
  $\eta$-expanded using the same variable name $y$, results into the
  same equation.  A further application of the flex-rigid rule solves
  the equation by $[X_1\mapsto \lambda x.\lambda y.a]$.

  In other cases, the resulting equation may be different.  The new
  rule transforms $\lambda x.X(x) \unif \lambda x.f(\lambda x.g(x))$
  into the equation $\lambda x.\lambda x.X_1(x) \unif \lambda
  x.\lambda x.g(x)$ and the substitution $[X\mapsto \lambda
  x.f(\lambda x.X_1(x))]$.  However, the original flex-rigid rule
  would give us $\lambda x.X_1(x) \unif \lambda x.\lambda x.g(x)$ and
  the substitution $[X\mapsto \lambda x.f(X_1(x))]$. In the subsequent
  $\eta$-expansion we can not use the name $x$, and we need a
  \emph{new name} $z$, and $\alpha$-conversion of the right-hand side
  getting $\lambda x.\lambda z.X_1(x,z) \unif \lambda x.\lambda
  z.g(z)$. Both equations are obviously distinct. However, to solve
  this second equation, $X_1$ can not use the first argument, because
  it is not used in the right-hand side.  Therefore, we can
  instantiate $X_1 \mapsto\lambda x.\lambda y.X_1'(y)$, and
  $\alpha$-convert the new variable name $z$, getting the same
  equation as with the new flex-rigid rule.
\end{example}

\begin{example}
  Given the equations $\big\{\lambda x.\lambda y. X(x,y) \unif
  \lambda x.\lambda y.f(\lambda x.Y(x,y)),$ $\lambda x.\lambda y.Z(x,y)
  \unif \lambda x.\lambda y.X(y,x)\big\}$, after solving the first
  equation and replacing $\big[X\mapsto \lambda x.\lambda y.f(\lambda
  x.Y(x,y))\big]$ into the second one, we get $\lambda x.\lambda y.Z(x,y)
  \unif \lambda x.\lambda y.\big((\lambda x.\lambda y.f(\lambda
  x.Y(x,y)))(y,x)\big)$.
  By Lemma~\ref{lem-beta}, we can $\beta$-reduce using swappings,
  instead of the usual standard substitution. The permutation will be
  $\Pi_2(\langle x,y\rangle,\langle
  y,x\rangle)=\Pi_1\big(\langle(x\,y)\mdot y\rangle,\langle
  x\rangle\big)\mdot(x\,y)=(x\,x)\mdot(x\,y) = (x\,y)$, and the result
  of the $\beta$-reduction will be
$$
\big(\lambda x.\lambda y.f\big(\lambda x.Y(x,y)\big)\big)(y,x) =_\beta
(x\,y)\mdot f\big(\lambda x.Y(x,y)\big) = f\big(\lambda y.Y(y,x)\big)
$$
\end{example}

\begin{lemma}
\label{lem-completeness-new-rule}
The algorithm described in Definition~\ref{def-new-rule} is sound and
complete and computes a most-general higher-order pattern unifier
whenever it exists, when names of free and bound variables are
disjoint.
\end{lemma}

\begin{proof}
  The algorithm computes basically the same most general unifiers than
  the Nipkow's algorithm.

  The fact that we use swapping instead of substitution to remove
  $\beta$-redexes is not a problem according to
  Lemma~\ref{lem-beta}. We will obtain a term that is
  $\alpha$-equivalent to the one that we would obtain with the
  traditional capture-avoiding substitution. Notice that in the lemma
  we require arguments of free variables (the sequence $\vec{y}$) to
  be a list of distinct bound variables. This is ensured in the case
  of higher-order pattern unification, but it is not true in the
  general $\lambda$-calculus. The algorithm preserves the disjointness
  of bound and free variable names. Therefore, the other condition of
  the lemma $\{\vec{y}\}\cap\FV(\lambda\vec{x}.t)$ is also satisfied.

  In the third $\alpha$-transformation rule, if $x\not\in\FV(t)$ and
  $x\in\FV(u)$ and the equation is solvable, then $x$ must occur in
  $u$ just below a free variable, as one of its arguments, and this
  free variable must be instantiated by a term that does not use this
  argument. Notice also that the three $\alpha$-transformation rules,
  when the equation is solvable, succeed in making the lists of most
  external $\lambda$-bindings equal in both sides of the equation.

  In the case of the flex-rigid rule, we may obtain an equation
  $\lambda \vec{x}.X_i(x_1,\dots,x_n) \unif \lambda \vec{x}.\lambda
  \vec{y}.u_i'$ that needs to be $\eta$-expanded, and where
  $\{x_1,\dots,x_n\}\cap\{\vec{y}\}\neq\emptyset$.  Let be
  $\{x_1',\dots,x_{n'}'\} = \{x_1,\dots,x_n\}\setminus\vec{y}$,
  i.e. the sequence of variables $x_i$'s not in $\vec{y}$.  In any
  solution of this equation $X_i$ can not use the variables of the
  intersection of $\{x_1,\dots,x_n\}\cap\{\vec{y}\}$. Therefore, we
  can extend the solution with $X_i\mapsto \lambda x_1,\dots
  x_n.\lambda\vec{y}.X_i'(x'_1,\dots,x'_{n},\vec{y})$, and get the
  equation $\lambda \vec{x}.\lambda
  \vec{y}.X_i'(x'_1,\dots,x'_{n'},\vec{y}) \unif \lambda
  \vec{x}.\lambda \vec{y}.u_i'$.

The flex-flex and rigid-rigid rules are the same as in Nipkow's algorithm.
\end{proof}

\begin{lemma}
\label{lem-prop-pattern} 
Let $P$ be a solvable pattern unification problem, where the set of
free and bound variable names are disjoint, and let $\langle
a_1,\dots,a_n\rangle$ be a list of the names of bound variables of the
problem. Then, there exists a most general unifier $\sigma$ such that
\begin{enumerate}
\item  $\sigma$ does not
  use other bound-variable names than the ones already used in the
  problem, i.e than $\{a_1,\dots,a_n\}$.
\end{enumerate}
If in the original problem all bound variables with the same name have
the same type, i.e. we have a type $\tau_i$ for every bound variable name
$a_i$, then
\begin{enumerate}
\setcounter{enumi}{1}
\item the same applies to $\sigma$, i.e. any bound variable of
  $\sigma$ with name $a_i$ has type $\tau_i$, and
\item any free variable $X$ occurring in $\sigma$ has type
  $\nu_1\to\cdots\to\nu_m\to\nu$, where $\langle
  \nu_1,\dots,\nu_m\rangle$ is a sublist of
  $\langle\tau_1,\dots,\tau_n\rangle$.
\end{enumerate}
\end{lemma}

\begin{proof}
  By Lemma~\ref{lem-completeness-new-rule} with the new transformation
  rules we obtain most general unifiers for solvable pattern
  unification problems. Then, by simple inspection of the new
  transformation rules, where all bound variable names in the
  right-hand sides of the rules are already present in the left-hand
  sides, we have that new equations and substitutions do not introduce
  new names.  In addition, since names of free and bound variables are
  distinct, $\beta$-reductions due to substitution applications
  satisfy conditions of Lemma~\ref{lem-beta}, therefore we can
  conclude that we do not need new bound variable names due to
  $\beta$-reductions either.

  Notice also that in these rules, when we introduce a new bound
  variable in the right-hand side, with a name already used in the
  left-hand side, both variables have the same type. When, we swap two
  variable names in an $\alpha$-conversion or in a $\beta$-reduction,
  they have also the same type.

  Finally, let $\sigma'$ be any most general unifier not using other bound
  variable names than the ones used in $P$, i.e. $a_1,\dots,a_n$. For
  every variable $X$ occurring free in $\sigma$, chose one of their
  occurrences. This will be of the form $X(b_1,\dots,b_m)$, where
  $\{b_1,\dots,b_m\}\subseteq \{a_1,\dots,a_n\}$ and the $b_i$'s are
  pairwise distinct. Let $\langle b_{\pi(1)},\dots,b_{\pi(m)}\rangle$
  be a sublist of $\langle a_1,\dots,a_n\rangle$. Then composing
  $\sigma'$ with $[X'\mapsto \lambda b_1.\cdots.\lambda
  b_m.X(b_{\pi(1)},\dots,b_{\pi(m)})]$, for every variable $X$, we get
  another most general unifier fulfilling the requirements of the
  third statement of the lemma. Notice that, although not all
  occurrences of $X$ have the same parameters, it does not matter
  which one we chose because all them have the same type.
\end{proof}

\section{The Reverse Translation}
\label{sec-reverse}

As we have shown, Theorem~\ref{the-trans-solution} is not enough to
prove that, if $\trans{\sf P}$ is solvable, then $\sf P$ is
solvable. We still have to prove that if $\trans{\sf P}$ is solvable,
then for some solution $\sigma$ of $\trans{\sf P}$ we can build a
nominal solution $\langle \nabla,\sigma'\rangle$ of $\sf P$. This is
the main objective of this section.  Taking into account that
$\trans{\sf P}$ is a higher-order pattern unification problem, and
that these problems are unitary, we will prove something stronger: if
$\trans{\sf P}$ is solvable, then $\trans{\sf \sigma}^{-1}$ is defined
for the most general unifier $\sigma$ of $\trans{\sf P}$. Moreover, in
the next section we will prove that $\trans{\sf \sigma}^{-1}$ is also
a most general nominal unifier.

\begin{definition}
Let $\sf\langle a_1,\dots,a_n\rangle$ be a fixed ordered list of atoms,
and let $\nabla$ be a freshness environment.  The back-translation
function is defined on $\lambda$-terms in $\eta$-long $\beta$-normal
form as follows: 
$$
\begin{array}{l}
\transb{a}_\nabla = \sf a\\[1mm]
\transb{f(t_1,\dots,t_n)}_\nabla = {\sf f}(\transb{t_1}_\nabla,\dots,\transb{t_n}_\nabla)\\[1mm]
\transb{\lambda a.t}_\nabla = {\sf a}\,.\,\transb{t}_\nabla \\[1mm]
\transb{X(c_1,\dots, c_m)}_\nabla = \sf 
\pi^{-1}
\mdot X\hspace{4mm} \begin{array}[t]{l}
          \mbox{where $\pi$ is a permutation on $\sf\langle
	   a_1,\dots,a_n\rangle$ satisfying}\\
	   \mbox{$\sf \langle \pi\mdot
	    c_1,\dots,\pi\mdot c_m\rangle$ is the sublist of $\sf\langle
	    a_1,\dots,a_n\rangle$ such that }\\
	  \mbox{$\sf\pi\mdot c_i\#
	    X\not\in\nabla$ and $\sf c_i$ and $\sf\pi\mdot c_i$ have the same sort}
          \end{array}
\end{array}
$$ 
where $a$ is a bound variable with name $\sf a$, $f$ is the constant
associated to the function symbol $\sf f$, either $X$ is the free
variable associated to $\sf X$, or if $X$ is a fresh variable then
$\sf X$ is a fresh nominal variable, and the permutation $\pi^{-1}$ is
supposed to be decomposed in terms of transpositions (swappings).
\end{definition}

Notice that the back-translation function is not defined for all
$\lambda$-terms, even for all higher-order patterns. In particular,
$\transb{\lambda x.t}$ is not defined when $x$ is not base typed, or
$\transb{x(t_1,\dots,t_n)}$ is not defined when $x$ is a bound
variable.

Notice also that the permutation $\pi$ is not completely determined by
the side condition of the forth equation. For instance, given $\langle
a_1,a_2,a_3\rangle$ as the list of atoms, all them of the same sort,
to define $\transb{X(a_1)}_{\{a_1\#X,a_2\#X\}} = \sf\pi^{-1}\mdot X$
the condition requires $\pi\mdot a_1 = a_3$, but then, we can choose
$\pi\mdot a_2 = a_1$ and $\pi\mdot a_3 = a_2$, or vice versa $\pi\mdot
a_2 = a_2$ and $\pi\mdot a_3 = a_1$. Therefore, $\transb{t}_\nabla$ is
nondeterministically defined.

For $\lambda$-substitutions the back-translation is defined as follows.

\begin{definition}
\label{def-back-substitution}
Let $\sf\langle a_1,\dots,a_n\rangle$ be a fixed ordered list of
atoms, and let $\nabla$ be a freshness environment. The
back-translation function is defined on $\lambda$-substitutions as
follows.
\[
\transb{\sigma}_\nabla =\bigcup_{X\in\dom(\sigma)} \Big[{\sf X}\mapsto \transb{\sigma(X)(a_1,\dots, a_n)}_\nabla\Big]
\]
\end{definition}

Notice that if $\sigma(X)(a_1,\dots, a_n)$ is not a well-typed
$\lambda$-term, or $\trans{\sigma(X)(a_1,\dots, a_n)}^{-1}_\nabla$ is not
defined for some $X\in\dom(\sigma)$, then $\trans{\sigma}^{-1}_\nabla$ is not
defined.

We introduce the following notion to describe which $\lambda$-terms
and substitutions have reverse translation w.r.t. a freshness
environment.

\begin{definition}\label{def-compatibility}
  Given a $\lambda$-term $t$ (resp.~$\lambda$-substitution $\sigma$),
  and a freshness environment $\nabla$, we say that $t$
  (resp.~$\sigma$) is $\nabla$-compatible if $\transb{t}_\nabla$
  (resp.~$\transb{\sigma}_\nabla$) is defined.
\end{definition}

\begin{lemma}\label{lem-inverse}
  For any $\lambda$-term $t$, and freshness environment $\nabla$, if
  $t$ is $\nabla$-compatible, then $\trans{\transb{t}_\nabla}_\nabla =
  t$.

\noindent
  For every $\lambda$-substitution $\sigma$, and freshness environment
  $\nabla$, if $\sigma$ is $\nabla$-compatible, then
  $\trans{\transb{\sigma}_\nabla}_\nabla=\sigma$.
\end{lemma}

\begin{proof}
  Let $\sf\langle a_1,\dots,a_n\rangle$ be a fixed ordered list of
  atoms.  The existence of $\transb{t}_\nabla$ restricts the form of
  $t$ to five cases. For the first four, the proof is trivial. In the
  case $t=X(c_1,\cdots,c_m)$, we have
$$
\begin{array}{rcl}
\trans{\transb{X(c_1,\cdots,c_m)}_\nabla}_\nabla &= &
\trans{\sf \pi^{-1}\mdot X}_\nabla\\
&=&
X\left(\trans{\sf \pi^{-1}\mdot\pi\mdot c_1}_\nabla,\cdots,\trans{\sf \pi^{-1}\mdot\pi\mdot c_m}_\nabla\right)\\
&=& X(c_1,\cdots,c_m)
\end{array}
$$
where $\pi$ is a permutation on $\sf\langle a_1,\dots,a_n\rangle$
satisfying $\sf \langle \pi\mdot c_1,\dots,\pi\mdot c_m\rangle$ is the
sublist of $\sf\langle a_1,\dots,a_n\rangle$ such that $\sf\pi\mdot
c_i\# X\not\in\nabla$ and $\sf c_i$ and $\sf\pi\mdot c_i$ have the same
sort.

For the second statement, by Definitions~\ref{def-back-substitution}
and~\ref{def-translation-solutions} we have
$$
\begin{array}{rcl}
\trans{\transb{\sigma}_\nabla}_\nabla &=&\displaystyle
\trans{\bigcup_{X\in\dom(\sigma)}[{\sf X}\mapsto \transb{\sigma(X)(a_1,\cdots, a_n)}_\nabla]}_\nabla\\[2mm]
&=&\displaystyle
\bigcup_{X\in\dom(\sigma)} \left[X\mapsto \lambda a_1\cdots a_n.\trans{\transb{\sigma(X)(a_1,\cdots, a_n)}_\nabla}_\nabla\right]\\
&=&\displaystyle \bigcup_{X\in\dom(\sigma)} [X\mapsto \lambda a_1\cdots a_n.\sigma(X)(a_1,\cdots, a_n)]\\
&=&\displaystyle \bigcup_{X\in\dom(\sigma)} [X\mapsto \sigma(X)] =\sigma
\end{array}
$$

Where we make use of the first statement to prove
$\trans{\transb{\sigma(X)(a_1,\cdots, a_n)}_\nabla}_\nabla =
\sigma(X)(a_1,\cdots, a_n)$. Notice that, if $\sigma$ is $\nabla$-compatible,
then $\sigma(X)(a_1,\cdots, a_n)$ is also  $\nabla$-compatible.
\end{proof}

Given a pattern unifier, in order to reconstruct the corresponding
nominal unifier, we have several degrees of freedom.  We start with
higher-order pattern unifier $\sigma$ with a restricted use of names
of bound variables. Then, we will construct a freshness environment
$\nabla$ such that $\sigma$ is $\nabla$-compatible.  This construction
is described in the proof of Lemma~\ref{lem-existence}, and it is
nondeterministic. The corresponding nominal solution is then
$\pair{\nabla}{\transb{\sigma}_\nabla}$. Moreover,
$\transb{\sigma}_\nabla$ is not uniquely defined.  The following
example illustrates these degrees of freedom in this back-translation.

\begin{example}
The nominal unification problem
$$
\sf P = \{a.a.X\unifn c.a.X\ ,\ a.b.X \unifn b.a.(a\,b)\mdot X\}
$$
where all atoms and variables are of the same sort, is translated as 
$$
\begin{array}{ll}
\trans{\sf P} = \{&\lambda a.\lambda b.\lambda c.\lambda a.\lambda a. X(a,b,c)
\unif
\lambda a.\lambda b.\lambda c.\lambda c.\lambda a. X(a,b,c)
\ ,\\
&\lambda a.\lambda b.\lambda c.\lambda a.\lambda b. X(a,b,c)
\unif
\lambda a.\lambda b.\lambda c.\lambda b.\lambda a. X(b,a,c)
\ \}
\end{array}
$$

Most general higher-order pattern unifiers are
$$
\sigma_1 = [X\mapsto \lambda a.\lambda b.\lambda c.  Z(a,b)]
$$
and
$$
\sigma_2 = [X\mapsto \lambda a.\lambda b.\lambda c.  Z(b,a)]
$$
which are equivalent.

Let $\sf\langle a,b,c\rangle$ be the fixed list of atoms.  Following
the construction described in the forthcoming proof of
Lemma~\ref{lem-existence}, for every variable $Z$ occurring in
$\sigma$, we construct a sublist of atoms $L_Z=\sf\langle
b_1,\dots,b_m\rangle$ satisfying $b_j:\transb{\tau_j}$, for every
$j=1,\dots,m$.  In our case, we can choose among three possibilities
$L_Z^1 = \sf\langle a,b\rangle$, $L_Z^2 = \sf\langle a,c\rangle$ or
$L_Z^3 = \sf\langle b,c\rangle$.  We construct $
\nabla=\bigcup_{\begin{array}{c}\mbox{\scriptsize $Z$ occurs in
      $\sigma$} \\[-2mm] \mbox{\scriptsize $ {\sf a\in\langle
        a_1,\dots,a_n}\rangle\setminus L_Z$}\end{array}} \{{\sf
  a\#Z}\} $.

From the two pattern unifiers $\sigma_i$'s, and the three lists $L_Z^j$'s
we can construct six possible nominal unifiers:

\[
\begin{array}{c|cc}
 & \sigma_1 & \sigma_2\\
\hline\\[-3mm]
L_Z^1 & 
\sf\langle \{c\# Z\}\ ,\ [X\mapsto\left(\begin{array}{ccc}\sf a&\sf
					b&\sf c\\\sf a&\sf b&\sf
					c\end{array}\right)^{-1}\mdot
Z]\rangle & 
\sf\langle \{c\# Z\}\ ,\ [X\mapsto\left(\begin{array}{ccc}\sf a&\sf
					b&\sf c\\\sf b&\sf a&\sf
					c\end{array}\right)^{-1}\mdot
Z]\rangle\\
L_Z^2 & 
\sf\langle \{b\# Z\}\ ,\ [X\mapsto\left(\begin{array}{ccc}\sf a&\sf
					b&\sf c\\\sf a&\sf c&\sf
					b\end{array}\right)^{-1}\mdot
Z]\rangle & 
\sf\langle \{b\# Z\}\ ,\ [X\mapsto\left(\begin{array}{ccc}\sf a&\sf
					b&\sf c\\\sf c&\sf a&\sf
					b\end{array}\right)^{-1}\mdot
Z]\rangle\\
L_Z^3 & 
\sf\langle \{a\# Z\}\ ,\ [X\mapsto\left(\begin{array}{ccc}\sf a&\sf
					b&\sf c\\\sf b&\sf c&\sf
					a\end{array}\right)^{-1}\mdot
Z]\rangle & 
\sf\langle \{a\# Z\}\ ,\ [X\mapsto\left(\begin{array}{ccc}\sf a&\sf
					b&\sf c\\\sf c&\sf b&\sf
					a\end{array}\right)^{-1}\mdot
Z]\rangle\\
\end{array}
\]

The permutations can be written as swappings obtaining:
\[\renewcommand{\arraystretch}{1.3}
\begin{array}{c|ll}
 & \multicolumn{1}{c}{\sigma_1} &  \multicolumn{1}{c}{\sigma_2}\\
\hline
L_Z^1 & 
\sf\langle \{c\# Z\}\ ,\ [X\mapsto Z]\rangle & 
\sf\langle \{c\# Z\}\ ,\ [X\mapsto (a\,b)\mdot Z]\rangle\\
L_Z^2 & 
\sf\langle \{b\# Z\}\ ,\ [X\mapsto (b\,c)\mdot Z]\rangle & 
\sf\langle \{b\# Z\}\ ,\ [X\mapsto (a\,b)(b\,c)\mdot Z]\rangle\\
L_Z^3 & 
\sf\langle \{a\# Z\}\ ,\ [X\mapsto (a\,c)(b\,c)\mdot Z]\rangle & 
\sf\langle \{a\# Z\}\ ,\ [X\mapsto (a\,c)\mdot Z]\rangle\\
\end{array}
\]

All these nominal unifiers are most general and equivalent. Notice that
these are \emph{all} the most general nominal unifiers.
\end{example}

\begin{lemma}\label{lem-existence}
  For every equational nominal unification problem $\sf P$, if the
  pattern unification problem $\trans{\sf P}$ is solvable, then there
  exists a freshness environment $\nabla$, and a most general pattern unifier
  $\sigma$, such that $\sigma$ is $\nabla$-compatible.
\end{lemma}

\begin{proof}
  The most general unifier $\sigma$ is chosen, accordingly to
  Lemma~\ref{lem-prop-pattern}, as a unifier not using other bound
  variable names than the ones used in $\trans{\sf P}$. Moreover,
  since all bound variables of $\trans{\sf P}$ with the same name
  $a_i$ have the same type $\tau_i$, the same happens in $\sigma$, and
  all free variables $Z$ occurring in $\sigma$ have a type of the form
  $Z:\tau_{i_1}\to\dots\to\tau_{i_m}\to\delta$, for some indexes
  satisfying $1\leq i_1 <\cdots <i_m\leq n$. Notice that there could
  be more than one set of indexes satisfying this condition.

  The freshness environment $\nabla$ is constructed as follows.  For
  any variable $Z:\tau_{i_1}\to \dots \tau_{i_m} \to \delta$
  occurring\footnote{We say that $X$ occurs in $\sigma$, if $X$ occurs
    free in $\sigma(Y)$, for some $Y\in\dom(\sigma)$.} in $\sigma$ ,
  let $L_Z=\sf\langle a_{i_1},\dots,a_{i_m}\rangle$ be a sublist of the atoms
  $\sf\langle a_1,\dots,a_n\rangle$. Then,
\[
\nabla=\hspace{-6mm}\bigcup_{\begin{array}{c}\mbox{\scriptsize $Z$ occurs in
      $\sigma$} \\[-2mm] \mbox{\scriptsize $ {\sf a\in\langle
        a_1,\dots,a_n}\rangle\setminus L_Z$}\end{array}}\hspace{-6mm} \{{\sf a\#Z}\}
\]

We prove that $\sigma(X)(a_1,\dots, a_n)$ is $\nabla$-compatible, for any
$X\in\dom(\sigma)$.

Since $\sigma$ is most general $\dom(\sigma)$ only contains variables
$X$ of $\trans{\sf P}$. All these variables have type
$\trans{\tau_1}\to\cdots\to\trans{\tau_n}\to\trans{\tau_0}$, where
$\langle\tau_1,\dots,\tau_n\rangle$ is the list of sorts of $\sf
\langle a_1,\dots,a_n\rangle$, and $\tau_0$ is the sort of $\sf X$.
Therefore, $\sigma(X)(a_1,\dots, a_n)$ is a well-typed $\lambda$-term.
Now we prove that this term is \emph{back-translatable} by structural
induction.

By Lemma~\ref{lem-prop-pattern}, $\sigma(X)$ does not use bound
variables with other names and types than the ones already used in the
original problem. This ensures that we can always translate back bound
variables $a$ as the atom with the same name $\sf a$. Terms formed by
a constant or free variable are particular cases of applications with
$m=0$, studied bellow.

All $\lambda$-abstractions will be of the form  $\lambda a_i.t$,
where $a_i=\trans{\sf a_i}$. This ensure that its translation back is
possible, if the body of the $\lambda$-abstractions is
back-translatable.

All applications are of the form $f(t_1,\dots,t_m)$ where $f$ is a
constant of the original nominal problem (since $\sigma$ is most
general), or of the form $X(a_{i_1},\dots,a_{i_m})$ where $X$ is a
free variable and $a_{i_1},\dots,a_{i_m}$ are distinct bound
variables. Notice that we can no have terms of the form
$a_i(t_1,\dots,t_n)$ where $a_i$ is a bound variable, because all
these bound variables have basic types. In the first case, the
application is back-translatable if arguments are. In the second case,
let $X:\tau_{j_1}\to\dots\to\tau_{j_m}\to\delta$, for some indexes
satisfying $1\leq j_1 <\cdots <j_m\leq n$.  using the $\nabla$
constructed before, we can translate back $X(a_{i_1},\dots,a_{i_m})$
as $\sf\pi^{-1}\mdot X$, for some $\pi$ satisfying $\sf\pi(a_{i_k}) =
a_{j_k}$, for $j=1,\dots,m$. Notice that $\sf a_{j_k}$
and $\sf a_{i_k}$ have the same sort $\tau_{j_k}$. Hence, this second
kind of applications is also back-translatable.
\end{proof}

\begin{theorem}\label{the-other-direction}
For every equational nominal unification problem $\sf P$, if the pattern
unification problem $\trans{\sf P}$ is solvable, then $\sf P$ is also
solvable.
\end{theorem}

\begin{proof}
  By Lemma~\ref{lem-existence}, if $\trans{\sf P}$ is solvable then
  there exist a most general unifier $\sigma$ of $\trans{\sf P}$, and
  a freshness environment $\nabla$ such that
  $\langle\nabla,\transb{\sigma}_\nabla\rangle$ is defined.
  W.l.o.g. assume that $\dom(\sigma)=\Vars(\trans{\sf P})$ and hence,
  according to Definition~\ref{def-back-substitution},
  $\dom(\transb{\sigma}_\nabla)=\Vars(\sf P)$. By
  Lemma~\ref{lem-inverse}, we have $\trans{\transb{\sigma}_\nabla}_\nabla=\sigma$, which solves
  $\trans{\sf P}$. Hence, by Theorem~\ref{the-trans-solution},
  $\langle\nabla,\transb{\sigma}_\nabla\rangle$ solves $\sf P$.
\end{proof}

From Theorems~\ref{the-one-direction} and~\ref{the-other-direction},
and linear-time decidability for Higher-Order Patterns
Unification~\cite{Qia96}, we conclude the following results.

\begin{corollary}
Nominal Unification is quadratic reducible to Higher-Order Pattern Unification.

\noindent
Nominal Unification can be decided in quadratic deterministic time.
\end{corollary}

\section{Correspondence Between Unifiers}
\label{sec-correspondence}

In this section we establish a correspondence between the solutions of
a nominal unification problem and their translations. We prove that
the translation function is monotone, in the sense that it translates
more general nominal solutions into more general pattern
solutions. The reverse translation also satisfies this
property. Therefore, both translate most general solutions into most
general solutions.  We start by generalizing the translation of a
nominal substitution w.r.t. a freshness environment, to respect the
translation of a nominal substitution w.r.t. \emph{two} freshness
environments, and similarly for the reverse translation.

\begin{definition}\label{def-generalized-inverse}
  Let $\sf\langle a_1,\dots,a_n\rangle$ be a fixed list of atoms.

  Given a nominal substitution $\sigma$, and two freshness
  environments $\nabla$ and $\nabla'$, satisfying $\sf\nabla\vdash
  \sigma(\nabla')$, we define
$$
\trans{\sigma}_{\nabla}^{\nabla'}= \bigcup_{\sf X\in\dom(\sigma)}[X\mapsto \lambda b_1,\dots, b_m.\trans{\sf\sigma(X)}_\nabla]
$$
where $\sf\langle b_1,\dots,b_m\rangle = \langle a\in\langle
a_1,\dots,a_n\rangle\ |\ a\#X\not\in\nabla'\rangle$.

 Given a pattern substitution $\sigma$, and two freshness
  environments $\nabla$ and $\nabla'$, we define
$$
{\transb{\sigma}}_{\nabla}^{\nabla'}= \bigcup_{X\in\dom(\sigma)}
\Big[{\sf X}\mapsto \transb{\sigma(X)(b_1,\dots, b_m)}_\nabla\Big]
$$
where $\sf\langle b_1,\dots,b_m\rangle = \langle a\in\langle
a_1,\dots,a_n\rangle\ |\ a\#X\not\in\nabla'\rangle$.

We say that $\sigma$ is $\nabla'\to\nabla$-compatible if
${\transb{\sigma}}_{\nabla}^{\nabla'}$ exists.
\end{definition}

Notice that this definition generalizes
Definition~\ref{def-translation-solutions} because
$\trans{\sigma}_\nabla=\trans{\sigma}_\nabla^\emptyset$, and
Definition~\ref{def-back-substitution} because,
$\transb{\sigma}_\nabla={\transb{\sigma}}_\nabla^\emptyset$.

The following lemmas are generalizations of
Lemmas~\ref{lem-technical2} and~\ref{lem-inverse},
respectively. Their proofs are also straightforward generalizations.

\begin{lemma}\label{lem-technical3}
  For any nominal substitution $\sf \sigma$, freshness environments
  $\nabla_1$ and $\nabla_2$, and nominal term $\sf t$, satisfying
  $\sf\nabla_2\vdash \sigma(\nabla_1)$ and
  $\sf\Vars(t)\subseteq\dom(\sigma)$, we have $$
\trans{\sf \sigma}_{\nabla_2}^{\nabla_1}(\trans{\sf t}_{\nabla_1}) = \trans{\sf\sigma(t)}_{\nabla_2}
$$
\end{lemma}
\begin{lemma}\label{lem-technical4}
For any $\lambda$-substitution $\sigma$ and freshness environment $\nabla_1$ and
$\nabla_2$, if $\sigma$ is $\nabla_1\to\nabla_2$-compatible, then
$$
\trans{{\transb{\sigma}}_{\nabla_2}^{\nabla_1}}_{\nabla_2}^{\nabla_1}=\sigma
$$
\end{lemma}

If a $\lambda$-substitution $\sigma_1$ is more general than another
$\sigma_2$, then there exists a substitution $\sigma_3$ that satisfies
$\sigma_2 = \sigma_3\circ\sigma_1$.  The following lemma states that
this substitution can be used to construct a nominal substitution
$\transb{\sigma_3}$ that we will use, in
Lemma~\ref{lem-monotonic-translation}, to prove that
$\transb{\sigma_1}$ is more general than $\transb{\sigma_2}$.

\begin{lemma}
\label{lem-sigma-compatibility}
For any pair of $\lambda$-substitutions $\sigma_1$ and $\sigma_2$ and freshness environments
$\nabla_1$ and $\nabla_2$, if
$\sigma_1$ is $\nabla_1$-compatible,
$\sigma_2$ is $\nabla_2$-compatible, and
$\sigma_1$ is more general than  $\sigma_2$, then there exists
a $\lambda$-substitution $\sigma_3$ such that
\begin{enumerate}
\item $\sigma_2 = \sigma_3\circ\sigma_1|_{\dom(\sigma_1)\cup\dom(\sigma_2)}$
\item $\sigma_3$ is $\nabla_1\to\nabla_2$-compatible, and
\item $\nabla_2\vdash \transb{\sigma_3}_{\nabla_2}^{\nabla_1}(\nabla_1)$.
\end{enumerate}
\end{lemma}

\begin{proof}
  The first conclusion is a consequence of $\sigma_1$ is more general
  than $\sigma_2$. However, w.l.o.g. we take a $\sigma_3$ that only
  instantiates variables occurring in $\sigma_1$ or belonging
  to $\dom(\sigma_2)$.

  For all $X\in\dom(\sigma_3)$, let $\sf\langle b_1,\dots,b_m\rangle =
  \langle a_i\ |\ a_i\#X\not\in\nabla_1\rangle$, where $\sf\langle
  a_1,\dots,a_n\rangle$ is the fixed list of atom names. Now, $X$
  occurs in $\sigma_1$ or $X\in\dom(\sigma_2)$.  In the first case,
  since $\sigma_1$ is $\nabla_1$-compatible and we are dealing with
  higher-order pattern substitutions, $X$ occurs in $\sigma_1$ in (at
  least one) subterm of the form $X(b'_1,\dots,b'_m)$, where $b_i'$
  are distinct bound variables with names in $\sf\langle
  a_1,\dots,a_n\rangle$, and $b_i$ and $b_i'$ have the same
  type. Moreover, $\sigma_3(X)(b_1',\dots,b_m')$, conveniently
  $\beta$-reduced, is a subterm of some $\sigma_2(Y)$, for some
  $Y\in\dom(\sigma_2)$. In the second case, if $X\in\dom(\sigma_2)$,
  we also have this property. Therefore, since $\sigma_2$ is
  $\nabla_2$-compatible, we have that
  $\sigma_3(X)(b_1',\dots,b_m')$, and hence
  $\sigma_3(X)(b_1,\dots,b_m)$ is $\nabla_2$-compatible.  Therefore,
  $\transb{\sigma_3}_{\nabla_2}^{\nabla_1}=\bigcup_{X\in\dom(\sigma_3)}[{\sf
    X}\mapsto\transb{\sigma_3(X)(b_1,\dots,b_m)}]$ exists, and
  $\sigma_3$ is $\nabla_1\to\nabla_2$-compatible.

  Let be $\sf b\#X\in\nabla_1$.  The free variable names of
  $\sigma_3(X)$ and $\sf\langle a_1,\dots,a_n\rangle$ are
  disjoint. Therefore, $b\not\in\FV(\sigma_3(X)(b_1,\dots,b_m))$,
  where $\sf\langle b_1,\dots,b_m\rangle = \langle a_i\ |\
  a_i\#X\not\in\nabla_1\rangle$. By Lemma~\ref{lem-inverse}, since
  $\sigma_3(X)(b_1,\dots,b_m)$ is $\nabla_2$-compatible, we have
  $b\not\in\FV\left(\trans{\transb{\sigma_3(X)(b_1,\dots,b_m)}_{\nabla_2}}_{\nabla_2}\right)$. By
  Lemma~\ref{lem-technical1}, ${\sf \nabla_2\vdash
    b}\#\transb{\sigma_3(X)(b_1,\dots,b_m)}_{\nabla_2}$.  By
  Definition~\ref{def-generalized-inverse}, ${\sf \nabla_2\vdash
    b}\#\transb{\sigma_3}_{\nabla_2}^{\nabla_1}({\sf X})$.  Therefore,
  we have $\sf
  \nabla_2\vdash\transb{\sigma_3}_{\nabla_2}^{\nabla_1}(\nabla_1)$.
\end{proof}

The following lemma ensures that the translation and reverse
translation of substitutions is monotone w.r.t. the more generality
relation.

\begin{lemma}\label{lem-monotonic-translation}
  For every nominal unification problem $\sf P$ and pair of unifiers
  $\langle\nabla_1,\sigma_1\rangle$ and
  $\langle\nabla_2,\sigma_2\rangle$, satisfying
  $\Vars({\sf P})\subseteq\dom(\sigma_1)\subseteq\dom(\sigma_2)$, we have
  $\pair{\nabla_1}{\sigma_1}$ is more general than
  $\pair{\nabla_2}{\sigma_2}$, if, and only if,
  $\trans{\sigma_1}_{\nabla_1}$ is more general than
  $\trans{\sigma_2}_{\nabla_2}$.
\end{lemma}

\begin{proof}\parindent 0mm
  $\Rightarrow$) By Theorem~\ref{the-trans-solution}, both
  $\trans{\sigma_1}_{\nabla_1}$ and
  $\trans{\sigma_2}_{\nabla_2}$ are solutions of
  $\sf\trans{P}$.
  If $\sf\langle\nabla_1,\sigma_1\rangle$ is more general than
  $\sf\langle\nabla_2,\sigma_2\rangle$, then there exists a nominal
  substitution $\sigma'$ such that $\sf\nabla_2\vdash \sigma'(\nabla_1)$
  and $\sf\nabla_2\vdash \sigma'\circ\sigma_1|_{\dom(\sigma_1)\cup\dom(\sigma_2)}\approx
  \sigma_2$. 
  For all $\sf X\in\dom(\sigma_2)$, we have $\nabla_2\vdash
  \sigma'(\sigma_1(X))\approx \sigma_2(X)$.
  By Lemma~\ref{lem-technical1},
  $\sf\trans{\sigma'(\sigma_1(X))}_{\nabla_2} =_\alpha
  \trans{\sigma_2(X)}_{\nabla_2}$.
  By Lemma~\ref{lem-technical3}, 
  $\sf\trans{\sigma'}_{\nabla_2}^{\nabla_1}(\trans{\sigma_1(X)}_{\nabla_1})
  =_\alpha \trans{\sigma_2(X)}_{\nabla_2}$.
  By Lemma~\ref{lem-technical2}, 
  $\sf\trans{\sigma'}_{\nabla_2}^{\nabla_1}(\trans{\sigma_1}_{\nabla_1}(\trans{X}_\emptyset))
  =_\alpha \trans{\sigma_2}_{\nabla_2}(\trans{X}_\emptyset)$.
  Since $\trans{\sf X}_\emptyset=X(a_1,\dots,a_n)$ and $a_i$ will be
  distinct free variables, we have
  $$
  \trans{\sigma_2}_{\nabla_2} (X) =
  \trans{\sigma'}_{\nabla_2}^{\nabla_1}\circ\trans{\sigma_1}_{\nabla_1}
  (X), \mbox{\hspace{10mm}for all
    $X\in\dom(\trans{\sigma_2}_{\nabla_2})$}
  $$
  Therefore, $\trans{\sigma_1}_{\nabla_1}$ is more general than
  $\trans{\sigma_2}_{\nabla_2}$.

  $\Leftarrow$) There exists a $\lambda$-substitution $\sigma'$ such
  that $\trans{\sigma_2}_{\nabla_2} =
  \sigma'\circ\trans{\sigma_1}_{\nabla_1}|_{\dom(\sigma_1)\cup\dom(\sigma_2)}$.
  By Lemma~\ref{lem-sigma-compatibility}, $\sigma'$ is
  $\nabla_1\to\nabla_2$-compatible.  Hence, it exists the nominal
  substitution $\sigma''=\transb{\sigma'}_{\nabla_2}^{\nabla_1}$.  For
  any $\sf X\in\dom(\sigma_2)$, by Lemmas~\ref{lem-technical3}
  and~\ref{lem-technical4}, we have
  $\trans{\sigma''\big(\sigma_1(X)\big)}_{\nabla_2} =
  \trans{\sigma''}_{\nabla_2}^{\nabla_1}\big(\trans{\sigma_1}_{\nabla_1}^{\emptyset}(\trans{\sf
    X}_\emptyset)\big) =
  \sigma'\big(\trans{\sigma_1}_{\nabla_1}^{\emptyset}(\trans{\sf
    X}_\emptyset)\big) =
  \trans{\sigma_2}_{\nabla_2}^{\emptyset}(\trans{\sf X}_\emptyset)=
  \trans{\sigma_2(\sf X)}_{\nabla_2}$. By Lemma~\ref{lem-technical1},
  we have $\sf \nabla_2\vdash \sigma''\big(\sigma_1(X)\big)
  \approx\sigma_2(X)$.  Therefore, $\sf \nabla_2\vdash
  \sigma''\circ\sigma_1|_{\dom(\sigma_1)\cup\dom(\sigma_2)}\approx
  \sigma_2$.  By Lemma~\ref{lem-sigma-compatibility}, we also have
  $\nabla_2\vdash \sigma''(\nabla_1)$.  From both facts, we conclude
  that $\sigma_1$ is more general than $\sigma_2$.
\end{proof}

\begin{corollary}
Most general nominal unifiers are unique.
\end{corollary}

\begin{proof}
It is a direct consequence of uniqueness of most general higher-order
pattern unifiers and Lemma~\ref{lem-monotonic-translation}.
\end{proof}

Finally we can conclude that the translations preserve most generality.

\begin{theorem}
  For any nominal problem $\sf P$ and nominal solution
  $\pair\nabla\sigma$, satisfying $\sf\Vars(P)\subseteq\dom(\sigma)$,
  $\pair\nabla\sigma$ is a most general unifier if, and only if,
  $\trans{\sigma}_\nabla$ is a most general unifier of
  $\sf\trans{P}$.
\end{theorem}

\begin{proof}
  $\Rightarrow$) Suppose that $\pair\nabla\sigma$ is a most general
  nominal unifier of $\sf P$, but $\trans{\sigma}_\nabla$ is not a
  most general pattern unifier of $\sf\trans{P}$.  By
  Theorem~\ref{the-trans-solution}, $\trans{\sigma}_\nabla$ is a
  solution of $\sf\trans{P}$. Since most general higher-order pattern
  unifiers are unique, and by Lemma~\ref{lem-existence}, there exists
  a most general pattern unifier $\sigma'$ of $\sf\trans{P}$
  \emph{strictly} more general than $\trans{\sigma}_\nabla$ and
  such that $\trans{\sigma'}^{-1}$ exists.  By
  Lemma~\ref{lem-inverse},
  $\trans{\trans{\sigma'}^{-1}}=\sigma'$. Since we assume that
  $\pair\nabla\sigma$ is most general and nominal most general
  unifiers are also unique, we have that $\pair\nabla\sigma$ is more
  general than $\trans{\sigma'}^{-1}$. Hence, by
  Lemma~\ref{lem-monotonic-translation}, $\trans{\sigma}_\nabla$
  is more general than $\trans{\trans{\sigma'}^{-1}}=\sigma'$, which
  contradicts that $\sigma'$ is strictly more general than
  $\trans{\sigma}_\nabla$.

  $\Leftarrow$) Suppose that $\trans{\sigma}_\nabla$ is most general,
  and $\pair\nabla\sigma$ is not.  Then, there exists a most general
  unifier $\pair{\nabla'}{\sigma'}$ such that $\pair{\nabla}{\sigma}$
  is not more general than $\pair{\nabla'}{\sigma'}$. On the other
  hand, since $\trans{\sigma}_\nabla$ is most general, it is more
  general than $\trans{\sigma'}_{\nabla'}$. Hence, by
  Lemma~\ref{lem-monotonic-translation}, $\pair\nabla\sigma$ is more
  general than $\pair{\nabla'}{\sigma'}$. This contradicts the initial
  assumption. Therefore, if $\trans{\sigma}_\nabla$ is most general,
  then $\pair\nabla\sigma$ must be most general.
\end{proof}

\section{Conclusions}
\label{sec-conclusions}

The paper describes a precise quadratic reduction from Nominal
Unification to Higher-Order Pattern Unification. This helps to better
understand the semantics of the nominal binding and permutations in
comparison with $\lambda$-binding and $\alpha$-conversion.  Moreover,
using the result of linear time decidability for Higher-Order Patterns
Unification~\cite{Qia96}, we prove that Nominal Unification can be decided in
quadratic time.

\bibliographystyle{acmtrans}

\end{document}